\documentclass[acmsmall,authorversion]{acmart}%
\setcopyright{rightsretained}
\acmPrice{}
\acmDOI{10.1145/3622841}
\acmYear{2023}
\copyrightyear{2023}
\acmSubmissionID{oopslab23main-p335-p}
\acmJournal{PACMPL}
\acmVolume{7}
\acmNumber{OOPSLA2}
\acmArticle{265}
\acmMonth{10}
\received{2023-04-14}
\received[accepted]{2023-08-27}

\usepackage{prelude}

\begin{document}


\title{A Grounded Conceptual Model for Ownership Types in Rust}

\author{Will Crichton}
\orcid{0000-0001-8639-6541}
\affiliation{
  \department{Department of Computer Science}
  \institution{Brown University}           
  \city{Providence}
  \state{Rhode Island}
  \postcode{02912}
  \country{USA}                    
}
\email{wcrichto@brown.edu}

\author{Gavin Gray}
\orcid{0000-0002-2960-1198}
\affiliation{
  \department{Department of Computer Science}
  \institution{ETH Z\"{u}rich}
  \city{Z\"{u}rich}
  \country{Switzerland}
}

\author{Shriram Krishnamurthi}
\orcid{0000-0001-5184-1975}
\affiliation{
  \department{Department of Computer Science}
  \institution{Brown University}           
  \city{Providence}
  \state{Rhode Island}
  \postcode{02912}
  \country{USA}                    
}

\begin{abstract}
Programmers learning Rust struggle to understand ownership types, Rust's core mechanism for ensuring memory safety without garbage collection. This paper describes our attempt to systematically design a pedagogy for ownership types. First, we studied Rust developers' misconceptions of ownership to create the Ownership Inventory, a new instrument for measuring a person's knowledge of ownership. We found that Rust learners could not connect Rust's static and dynamic semantics, such as determining why an ill-typed program would (or would not) exhibit undefined behavior. Second, we created a conceptual model of Rust's semantics that explains borrow checking in terms of flow-sensitive permissions on paths into memory. Third, we implemented a Rust compiler plugin that visualizes programs under the model. Fourth, we integrated the permissions model and visualizations into a broader pedagogy of ownership by writing a new ownership chapter for \textit{The Rust Programming Language}, a popular Rust textbook. Fifth, we evaluated an initial deployment of our pedagogy against the original version, using reader responses to the Ownership Inventory as a point of comparison. Thus far, the new pedagogy has improved learner scores on the Ownership Inventory by an average of \data{9\%} ($N = \data{342}, d = \data{0.56}$).
\end{abstract}



\maketitle

\section{Introduction}
\label{sec:intro}

%
%
%
%
%

Ownership is a programming discipline for managing the aliasing and mutation of data, enforced statically through ownership types. The flagship programming language for ownership is Rust, which empowers programmers to write memory-safe code without garbage collection. Rust's ownership model synthesizes several ideas from PL research such as linear logic\,\cite{girard1987linear}, class-based alias management\,\cite{clarke1998ownership}, and region-based memory management\,\cite{grossman2002region}. History shows that developers cannot write memory-safe C and C++ in practice\,\cite{msrc2019}, so the software industry is turning toward Rust. For example, Google's Android team has found zero memory vulnerabilities in 1.5 million lines of Rust code\,\cite{android2022}.

This rosy picture of PL tech transfer belies a persistent obstacle: teaching ownership types to prospective users. Over the last four years, studies have found that Rust learners struggle to fix ownership type errors\,\cite{zeng2018,zhu2022}, and users self-report that ownership is among their biggest barriers to learning Rust\,\cite{rustsurvey2020,fulton2021}. To wit: advances in the technical factors of type systems require commensurate advances in the human factors of type systems.

Our work started with the question: how can we systematically design a pedagogy for ownership types? Today, popular pedagogies for advanced type systems are driven by experts' intuitions about how people learn, as well as by intuitions about what makes type systems difficult to understand. 
As practicing (computer) scientists, we wanted to approach pedagogic design through scientific principles: grounding the pedagogy in observations about the struggles of Rust learners, and then evaluating the pedagogy by its effects on learning outcomes. This paper describes how we put these principles into practice:

\begin{enumerate}
    \item \textbf{We ran a formative user study to identify misconceptions about ownership types held by Rust learners} (\Cref{sec:ci}).  
    We designed a new instrument for evaluating understanding of ownership, the Ownership Inventory, by drawing tasks from commonly-reported Rust issues on StackOverflow. We studied $N = \data{36}$ Rust learners trying to solve Ownership Inventory problems. We found that learners can generally identify the surface-level reason for why a program is ill-typed with respect to ownership. However, they do not understand what undefined behavior (if any) would occur if an ill-typed program were executed. This misunderstanding is reflected in often inefficient and incorrect strategies that participants used to fix ownership errors.

    \item \textbf{We developed a conceptual model of ownership types to address these misconceptions} (Sections \ref{sec:dynformalmodel} and \ref{sec:statformalmodel}).
    The conceptual model represents the aspects of Rust's static and dynamic semantics that are relevant to ownership, while abstracting other details. The model provides learners a foundation to understand essential concepts such as undefined behavior and the incompleteness of Rust's ownership type-checker, or ``borrow checker.''

    \item \textbf{We implemented tools to visualize Rust programs under the conceptual model} (Sections \ref{sec:dynimpl} and \ref{sec:statimpl}).
    We leveraged existing executable models of Rust's dynamic  and static semantics to generate traces for a given Rust program. We visualize these traces to help learners ``see'' the abstract conceptual model reified into concrete examples.

    \item \textbf{We designed a pedagogy around our conceptual model to explain ownership types} (\Cref{sec:pedagogy}).
    We wrote a chapter on ownership that explains how and why Rust's type system prevents undefined behavior, illustrating code with diagrams generated by our tool. We integrated this chapter into a popular Rust textbook, \textit{The Rust Programming Language} (\trpl{})\,\cite{trpl}. We setup and advertised a public website that hosts our \trpl{} fork.

    \item \textbf{We A/B tested our pedagogy against the \trpl{} baseline} (Section \ref{sec:eval}).
    We measured learning outcomes with two kinds of quizzes: simpler comprehension questions about the conceptual model, and more difficult multiple-choice versions of the Ownership Inventory.
    Learners could correctly answer comprehension questions with \data{72\%} accuracy. Our initial deployment improved the average Inventory score from \data{48\%} to \data{57\%} ($N = \data{342}, p < \data{0.001}, d = \data{0.56}$).
\end{enumerate}

The thesis of our pedagogy is that to understand ownership types, Rust learners need to understand two key concepts:  undefined behavior 
and incompleteness. What are the ``stuck states'' of Rust's dynamic semantics? Why does Rust's static semantics avoid stuck states? What valid programs are rejected by the static semantics? Existing Rust pedagogies accurately characterize the syntactic properties enforced by the compiler. However, they fall short of explaining the counterfactual \ub{} that could occur without the borrow checker, especially regarding memory safety. For example, three popular Rust books explain mutable references like this:
\begin{itemize}
    \item \textit{Rust in Action}\,\cite{McNamara2021-ot}: ``Borrows can be read-only or read-write. Only one read-write borrow can exist at any one time.'' (No justification is provided.)
    \item \textit{Hands-on Rust}\,\cite{Wolverson2021-vz}: ``Rust's safety features make a mutable borrow \emph{exclusive}. If a variable is mutably borrowed, it cannot be borrowed---mutably or immutably---by other statements while the borrow remains.'' (No justification is provided.)
    \item \textit{The Rust Programming Language}\,\cite{trpl}: ``If you have a mutable reference to a value, you can have no other references to that value. [...] The benefit of having this restriction is that Rust can prevent data races at compile time.  [...] Users of an immutable reference don’t expect the value to suddenly change out from under them!''
\end{itemize}

\noindent Of these resources, only \trpl{} provides a sense of the counterfactual by abstractly gesturing towards the issue of data races. But learners likely need to see concrete ``negative'' examples of counterfactual behavior to better grasp the overall logic of ownership\,\cite{dyer2022negative}. Furthermore, none of these resources explain incompleteness, or even suggest that safe programs can be rejected by the compiler. Programmers need different strategies to fix ownership errors if a program is rejected due to incompleteness\,\cite{crichton2022ownershipusability}.

To address these issues, we designed our pedagogy around a new conceptual model of Rust's dynamic and static semantics. To explain \ub{}, we ``turn off'' Rust's  borrow checker to interpret and visualize programs that would otherwise be rejected by the compiler. These counterfactual visualizations help learners understand \ub{} that is avoided by the compiler. To explain incompleteness, we reframe ownership type-checking as a form of abstract interpretation. The abstract state of the program is a mapping from paths to permissions such as ``readable'' or ``writable.'' We visualize these abstract permission states to show learners how incompleteness arises from phenomena such as field-insensitivity and limitations of the alias analysis within the borrow checker.

\section{A Concept Inventory for Ownership}
\label{sec:ci}

To develop a pedagogy for ownership types, we first sought to understand the core misconceptions that underlie the struggles of Rust learners. Many learning resources are developed based on educators' intuitions about what makes concepts difficult. Instead, we sought to ground our pedagogy in experimental data about the experiences of Rust learners. Our overarching methodology was the development of a \emph{concept inventory} for ownership, henceforth called the ``Ownership Inventory.''

\subsection{Background}

In education research, a concept inventory (CI) is a test, usually composed of multiple-choice questions, about a narrow domain where the questions and distractors are drawn from common misconceptions about the domain\,\cite{hestenes1992fci}. There is no singular method for devising a CI\,\cite{lindell2007}, but the general idea is to articulate a range of important concepts in the target domain, and then to elicit misconceptions from learners about those concepts. The CI helps evaluate curricula for whether they effectively address such misconceptions.

CIs are an increasingly popular tool for CS education researchers. The last decade has seen a Cambrian explosion of CIs for CS topics such as 
CS1\,\cite{kaczmarczyk2010}, 
CS2\,\cite{wittie2017},
algorithms\,\cite{farghally2017}, recursion\,\cite{hamouda2017}, 
data structures\,\cite{porter2019}, 
digital logic\,\cite{herman2010}, 
operating systems\,\cite{webb2014}, and cybersecurity\,\cite{poulsen2021}. The development of these CIs has proved useful in demonstrating both the existence and frequency of misconceptions. For example, the Java-focused CI of \citet{kaczmarczyk2010} revealed that many students who completed a CS1 course ended up with a ``dearth of even basic conception of an Object.'' As another example, \citet{farghally2017} developed an algorithms CI based on 17 misconceptions predicted by educators. After administering the CI to two iterations of the algorithms course at their university, they found that 7 misconceptions were held by at least 1/3 of students. This kind of data helps guide instructors in determining which misconceptions are most important to address in curricular development.

\subsection{Development}

\begin{figure}
    \centering
    \begin{subfigure}{0.46\linewidth}
    \begin{tcolorbox}[boxsep=0pt,left=0pt,right=0pt,top=0pt,bottom=0pt,colframe=excerptframe]
    \includegraphics[width=\linewidth]{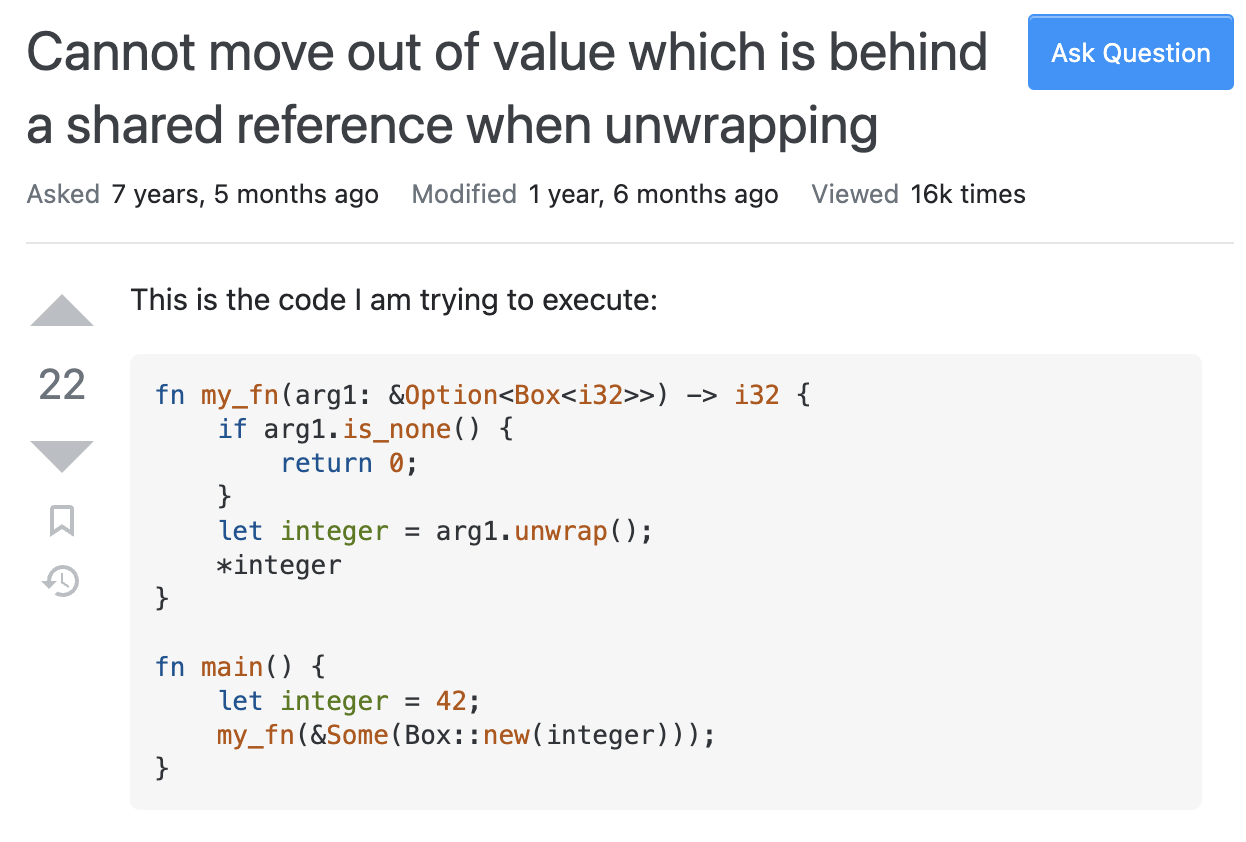}
    \end{tcolorbox}
    \caption{StackOverflow question \#32338659 that asks about moving out of a shared reference (a form of illegal borrow promotion).}
    \end{subfigure}
    \hfill
    \begin{subfigure}{0.49\linewidth}
\begin{lstlisting}
/// Gets the string out of an option if
/// exists, returning a default otherwise
fn get_or_default(arg: &Option<String>) 
  -> String 
{
    if arg.is_none() {
        return String::new();
    }	
    let s = arg.unwrap();
    s.clone()
}
\end{lstlisting}
    \cprotect\caption{The program we adapted from the question for the ownership inventory. The type \rs{Box<i32>} has been simplified to \rs{String}, the function has been given a meaningful name, and a doc-string has been added.}
    \end{subfigure}
    \cprotect\caption{An example of how we created snippets for the Ownership Inventory.}
    \label{fig:so-questions}
    \vspace{0.5em}
\end{figure}

\noindent We set out to design the Ownership Inventory for two reasons. First, the misconceptions we would observe in creating the Inventory would inform our later pedagogy. Second, we could use the Inventory to evaluate the efficacy of an intervention. If our pedagogy is better than before, then it should cause learners to score higher on the Inventory.

To construct the Inventory, we designed open-ended questions about ownership in situations that frequently stymie Rust learners. As we will describe in \Cref{sec:user-study-method}, we invited Rust learners to answer these questions, and then qualitatively analyzed their responses to characterize their misconceptions. Finally, we converted each open-ended question to multiple-choice by turning common misconceptions into distractors.


This method presents a chicken-and-egg problem: how do we know which situations are hard for Rust learners until we study them? So we turned to the world's largest repository of programmer struggles: StackOverflow. We searched for the most common questions asked about Rust on StackOverflow that pertain to ownership. Specifically, we queried the top \data{50} ``Most Frequent'' questions with the \verb|[rust]| tag and manually filtered the list to \data{27} questions involving ownership. We iteratively categorized each question, and identified four main categories of ownership problems:
\begin{itemize}
    \item \textbf{Dangling pointers:} references to stack-allocated values that escape their scope.
    \item \textbf{Overlapping borrows:} mutating data that is aliased by another reference.
    \item \textbf{Illegal borrow promotion:} writing to read-only data or moving borrowed data.
    \item \textbf{Lifetime parameters:} taking multiple references as input and returning references as output, but failing to specify the relationship between these references' lifetimes.
\end{itemize}

\noindent For each category, we selected a few representative StackOverflow questions and cleaned up the snippet in question. For example, \Cref{fig:so-questions} shows a StackOverflow question in the ``illegal borrow promotion'' category and the corresponding clean program. This process created eight total programs. The full set of programs is provided in \Cref{sec:all-inventory-snippets}.

\begin{figure}
\begin{tightexcerpt}
\begin{small}
\begin{enumerate}[itemsep=0.5em]
    \item The following Rust function is rejected by the Rust compiler. What error message would you expect from the compiler? \\ 
    {\footnotesize (You do not need to exactly reproduce the wording --- the question is about whether you generally understand how Rust would justify rejecting this function.)}
    
    \item Assume that the compiler did NOT reject this function. 
    \begin{enumerate}
        \item What is a program that calls this function which would violate memory safety or cause a data race? \\
        {\footnotesize (If no such program exists, then leave this field blank and explain your reasoning below. If you are uncertain of a particular Rust syntax, you may use pseudocode notation.)}
        \item In a few sentences, explain why you believe your program will violate memory safety or cause a data race, or why it is impossible to write such a program.
    \end{enumerate}
    
    \item 
    \begin{enumerate}
        \item How can this function be changed to pass the compiler while (1) preserving its intent and (2) minimally impacting runtime performance? \\ 
        {\footnotesize (You may use the standard library documentation and Rust compiler for this task. There is no right answer --- use your judgment. You can change any aspect of the function, including the type signature.)}
        \item In a few sentences, explain why your fix satisfies the criteria above.
    \end{enumerate}
\end{enumerate}
\end{small}
\end{tightexcerpt}
\caption{The template for open-ended Ownership Inventory questions about a given program that is rejected by the compiler. The small parenthetical text is provided as additional context.}
\label{fig:inventory-question-template}
\vspace{0.5em}
\end{figure}

We designed a single set of template questions that apply to each program. The template represents each stage of reasoning involved in fixing an ownership error. \Cref{fig:inventory-question-template} shows the exact wording of the questions. The phrasing of the questions is open-ended to elicit misconceptions without biasing respondents towards preconceived incorrect answers.

\subsection{Methodology}
\label{sec:user-study-method}

After developing the open-ended Ownership Inventory, we next administered the Inventory to elicit misconceptions that Rust learners have about ownership.

\subsubsection{Participants}

We recruited $N = \data{36}$ participants for the study. We found Rust learners by embedding an advertisement for the study within the online version of \trpl{}. Participants were required to be 18 years or older, and they were required to have completed reading \trpl{} before participating. Participants were compensated \$20. Participants on average had $\data{1.7}$ ($\sigma = \data{2.2}$) prior years of experience with either C or C++.

\subsubsection{Materials}
\label{sec:ci-materials}


\noindent We created a web interface that presents participants with a program and prompts for open-ended responses to each question in \Cref{fig:inventory-question-template}.
The interface uses the Monaco code editor running a Rust language server via a WebAssembly build of Rust Analyzer. The in-browser IDE allows participants to get information about the type and functionality of unfamiliar methods.

The materials include a tutorial that guides participants through both the technical details of using the interface, as well as a sample program paired with a sample response to each of the questions. 
The full source code for the experiment is provided in the artifact.

An important caveat: in our initial materials, participants were instructed to answer questions Q2b and Q3b by writing a code comment in the same editor used for questions Q2a and Q3a, respectively. However, after reviewing data from the first half of the experiment (\data{18} participants), we found that most participants either ignored or forgot this instruction. For the next \data{18} participants, we modified the website such that Q2b/Q3b had separate text boxes, which succeeded in eliciting responses from participants. Consequently, our results contain fewer data points for Q2b/Q3b than the other questions.

\subsubsection{Procedure}

Participants provided informed consent via the web interface, filled out information about their programming background, and then followed the tutorial. After learning about the style of task, participants were given three randomly selected programs in a random order. Participants had up to 15 minutes to answer all questions for a given program. Participants could complete the experiment at any time, and were not supervised by research staff.

At the end of the experiment, participants were given the option to provide open-ended feedback on their experience in the experiment. During the experiment, we continually monitored feedback for confusions with the materials. We did not ultimately make any changes based directly on participant feedback, which was most commonly of the form ``I found some of the questions difficult.''

\subsubsection{Analysis}

\noindent To evaluate the overall accuracy of participants, the first two authors independently coded each response as correct or incorrect. 
%
%
After the first round of coding, the authors resolved major disagreements, then independently re-coded the data. After the second round, the inter-rater reliability was 91\% in terms of raw agreement and $\kappa = \data{0.81}$ as measured by Cohen's $\kappa$, which is generally considered ``excellent''\,\cite{fleiss2013kappa} or ``almost perfect''\,\cite{landis1977kappa} agreement. We considered this sufficient agreement to proceed with the analysis. For the quantitative results, we report scores as the average of the two raters' scores on each item. 

To characterize the specific misconceptions that led to incorrect answers, we performed a thematic analysis of participant responses. The first author coded each incorrect response for the category of error displayed in the response, such as ``changing a type from \rs{&Option<String>} to \rs{Option<String>}.'' Error categories were further categorized based on similarities across problems, such as ``using \rs{clone} to satisfy the borrow checker.''

\subsection{Results}

\begin{table}
    \small
    \centering
    \begin{tabular}{l|l|rr|rr|rr|rr|rr}
    \textbf{Program} & \textbf{Category} & \textbf{Q1} & $\boldsymbol{N}$ &  \textbf{Q2a} & $\boldsymbol{N}$ &  \textbf{Q2b} & $\boldsymbol{N}$ &  \textbf{Q3a} & $\boldsymbol{N}$ &  \textbf{Q3b} & $\boldsymbol{N}$ \\
    \midrule
 {\footnotesize \verb|make_separator|} &         Dangling pointer & 77\% &                  15 & 50\% &                 15 & 35\% &                              10 & 80\% &             15 & 65\% &                          10 \\ \hline
  {\footnotesize \verb|get_or_default|} & Illegal borrow promotion & 80\% &                  15 & 73\% &                 15 &  0\% &                               4 & 80\% &             15 &  0\% &                           4 \\ \hline
    {\footnotesize \verb|remove_zeros|} &      Overlapping borrows & 85\% &                  13 & 54\% &                 13 & 25\% &                               4 & 42\% &             13 & 75\% &                           4 \\ \hline
         {\footnotesize \verb|reverse|} &      Overlapping borrows & 75\% &                  16 & 22\% &                 16 & 10\% &                              10 & 72\% &             16 & 45\% &                          10 \\ \hline
        {\footnotesize \verb|find_nth|} & Illegal borrow promotion & 88\% &                  17 & 15\% &                 17 &  5\% &                              11 & 12\% &             17 & 18\% &                          11 \\ \hline
     {\footnotesize \verb|apply_curve|} &      Overlapping borrows & 60\% &                  15 & 30\% &                 15 & 28\% &                               9 & 37\% &             15 & 33\% &                           9 \\ \hline
      {\footnotesize \verb|concat_all|} &      Lifetime parameters & 10\% &                  15 &  3\% &                 15 & 11\% &                               9 & 43\% &             15 & 11\% &                           9 \\ \hline
 {\footnotesize \verb|add_displayable|} &      Lifetime parameters & 38\% &                  16 &  6\% &                 16 &  5\% &                              10 &  6\% &             16 & 10\% &                          10 \\ \hline
    \multicolumn{1}{l}{} & 
    \multicolumn{1}{r|}{\textbf{Average accuracy:}} &
    \textbf{64\%} & & \textbf{31\%} & & \textbf{15\%} & & \textbf{46\%} & & \textbf{31\%}
    \end{tabular}
    \vspace{0.5em}
    \caption{Percentage of correct responses by program and question for the open-ended Ownership Inventory.}
    \label{tab:open-ended-data}
    \vspace{-2em}
\end{table}

\noindent\Cref{tab:open-ended-data} shows the percentage of total correct responses per-question. Participants could usually predict why the borrow checker would reject a program (Q1). However, participants could only fix the program in \data{46\%} of cases (Q3a), and could only create a counterexample in \data{31\%} of cases (Q2a). Their accuracy further drops when asked to justify their answer (Q2b and Q3b). Participants could sometimes create counterexamples and fixes without understanding why their answer is correct.


\subsubsection{Misconceptions about undefined behavior}

Participants' reasonable performance on Q1 suggests that Rust learners generally understand the surface-level reason for why a program is rejected. However, participants' comparatively poor performance on Q2a and Q2b suggests that Rust learners do not understand the deeper reasons that justify the ownership rules. Participants' incorrect attempts to construct counterexamples reveal a range of misconceptions about \ub{}.

\begin{wrapfigure}{r}{0.45\linewidth}
\begin{lstlisting}
fn make_separator(user_str: &str) -> &str 
{
  if user_str == "" {
    let default = "=".repeat(10);
    &default
  } else {
    user_str
  }
}    
\end{lstlisting}
\vspace{-0.5em}
\end{wrapfigure}

\noindent Participants frequently struggled to construct a correct counterexample to an unsafe function. For example, consider the \rs{make_separator} program, shown on the right, which returns a dangling pointer to the variable \rs{default}.  \data{5/15} responses to Q2a provided a counterexample that called  \rs{make_separator("")} but did not use the output of the function call. These participants considered calling this function sufficient to violate memory safety, not realizing the need to actually use the dangling pointer.

\begin{wrapfigure}{l}{0.36\textwidth}
\vspace{-0.5em}
\begin{lstlisting}
fn add_displayable<T: Display>(
  v: &mut Vec<Box<dyn Display>>, 
  t: T
) {
  v.push(Box::new(t));
}
\end{lstlisting}
\end{wrapfigure}

\noindent As a more complex example, consider the \rs{add_displayable} program (left) which lacks correct lifetime parameters. If the type \rs{T} contains a reference, then the lifetime of that reference is erased when converting \rs{T} into a trait object \rs{Box<dyn Display>}. Without specifying how \rs{T} relates to \rs{dyn Display}, that erasure allows  \rs{v} to outlive references in its elements. \\[-1em]

\begin{wrapfigure}{r}{0.52\textwidth}
\vspace{-1em}
\begin{lstlisting}
let mut v = Vec::new();
{
  let some_string = String::from("Some string");
  add_displayable(&mut v, some_string);
}
dbg!(v.last().unwrap());

\end{lstlisting}
\vspace{-1em}
\end{wrapfigure}

\noindent \data{No participants} managed to write a correct counterexample for this task. The answer closest to correct is shown on the right (\data{4/15} participants gave a comparable answer). This answer is not a counterexample because \rs{add_displayable} moves the input string into the vector, so no data is deallocated upon exiting the nested scope. This snippet becomes a correct counterexample if a reference is added to \rs{v}, e.g. the function call is changed to \rs{add_displayable(&mut v, &some_string)}. Then \rs{v} contains a dangling pointer to \rs{some_string}, and a read of that pointer would be \ub{}.

\begin{wrapfigure}{l}{0.34\linewidth}
\vspace{-0.5em}
\begin{lstlisting}
fn reverse(v: &mut Vec<i32>) {
  let n = v.len();
  for i in 0 .. n / 2 {
    std::mem::swap(
      &mut v[i], 
      &mut v[n - i - 1]
    );
  }
}
\end{lstlisting}
\vspace{-1em}
\end{wrapfigure}

Participants also struggled to identify when a function is actually safe and no counterexample exists. Consider the \rs{reverse} program (left), a case of overlapping borrows. Rust considers \rs{&mut v[i]} and \rs{&mut v[n - i - 1]} to possibly alias. But $i \neq n - i - 1$ for $i \in [0, n/2)$ so this program is actually safe. Only \data{3/15} participants identified this fact, and only \data{1} of those \data{3} gave a correct justification. 
Participant performance was comparably poor on Q2a and Q2b for the other two safe programs, \rs{find_nth} (\data{6\%} / \data{0\%}) and \rs{apply_curve} (\data{29\%} / \data{25\%}). \\[-1.5em]

\subsubsection{Misconceptions about fixing ownership errors}
\label{sec:misconceptions-fixing}

Participants could usually change a broken function to pass the borrow checker, but these fixes were not always correct and idiomatic.

\begin{wrapfigure}{r}{0.41\linewidth}
\vspace{-1.5em}
\begin{lstlisting}
fn reverse(v: &mut Vec<i32>) {
  let n = v.len();
  let mut v2 = v.clone();
  for i in 0 .. n / 2 {
    std::mem::swap(
      &mut v[i], &mut v2[n - i - 1]);
  }
}
\end{lstlisting}
\vspace{-0.5em}
\end{wrapfigure}

One common strategy we observed is the use of the \rs{.clone()} method. In Rust, cloning avoids aliasing by creating a deep copy of data. However, participants often incorrectly used \rs{clone}. As shown on the right, when fixing the \rs{reverse} program, \data{2/16} participants avoided overlapping borrows by cloning the input vector \rs{v}, and then swapping between the two vectors. This ``fix'' only reverses the first half of the input vector.

\begin{wrapfigure}{l}{0.45\textwidth}
\begin{lstlisting}
fn add_displayable<T: Display + 'static>(
  v: &mut Vec<Box<dyn Display>>, 
  t: T
);

fn add_displayable<'a, T: Display + 'a>(
  v: &mut Vec<Box<dyn Display + 'a>>, 
  t: T
);
\end{lstlisting}
\vspace{-1em}
\end{wrapfigure}    

When fixes required editing the type signature of a function, participants often created type signatures that were too restrictive or non-idiomatic. For example, returning to the \rs{add_displayable} program, \data{12/16} participants fixed the function by adding the \rs{+ 'static} bound to the generic type \rs{T}, as shown on the left on lines 1-4. (Notably, this solution is suggested in the compiler error for the original function.) However, this type signature is unnecessarily restrictive --- it prevents \rs{add_displayable} from being used on any type containing a non-static reference. An idiomatic solution is shown on lines 6-9, where a lifetime parameter \rs{'a} is added to the bounds of both \rs{T} and the trait object \rs{dyn Display} to indicate the aliasing relationship between the two types. Only \data{one} participant provided the correct and idiomatic solution.

\subsection{Discussion}

These results show that participants were quite capable at understanding the surface rules of ownership types. Excluding the two questions about lifetime parameters, participants could correctly predict the compiler's reason for rejection in 78\% of cases. However, the subsequent questions reveal that this understanding is shallow. On average, participants could not construct counterexamples to demonstrate \ub{}, nor could they effectively fix an ownership error.

Given these results, the key question for pedagogy is: why do Rust learning resources like \trpl{} lead to these learning outcomes? Learning is a complex process, so it is difficult to point to a specific passage and say, ``this is the problem.'' But in light of the misconceptions observed during this study, we hypothesized that a major learning challenge is that \trpl{} does not provide the foundations to reason counterfactually about \ub{}. Nor does \trpl{} explain how the borrow checker actually works, especially with respect to soundness vs. completeness.

\section{A Conceptual Model for Ownership}
\label{sec:notionalmachines}

To understand \ub{} and incompleteness, a person needs to understand Rust's dynamic and static semantics. At least, they need to understand a \emph{way of thinking} about these semantics that is \emph{viable} for tasks like debugging ownership errors. The responses to the open-ended Ownership Inventory showed that our participants had a fragile mental model of Rust's semantics. Therefore, we designed a new way of thinking (a ``conceptual'' model) of Rust that is precise enough to explain the relevant aspects of ownership, but approximate enough to avoid unnecessary detail.

To characterize the tension between precise and approximate conceptual models, consider a person who wants to understand integer addition in Rust. That is, they want a conceptual model of the semantics of the statement \rs{let z = x + y} where \rs{x, y : i32}. The true dynamic semantics of Rust's integer addition include aspects like two's complement overflow and auto-vectorization due to LLVM's optimizations. But for the average Rust user, these details are usually irrelevant for correctly using addition in routine programming tasks. A mental model that approximates the semantics as ``$x, y \in \mathbb{Z}$ and $x + y$ is standard integer addition'' is a generally viable model.

In this section, we describe our conceptual model of Rust's semantics designed to give learners a viable understanding of ownership. We provide a model for Rust's dynamic semantics (\Cref{sec:dynsem}) and for Rust's static semantics (\Cref{sec:statsem}). For each model, we articulate three aspects. First, the \textit{informal model}, an intuitive and visual representation of the model, as it would be presented to a Rust learner. Second, the \textit{formal model}, a precise and logical representation of the model for communication and reasoning within this paper. And third, the \textit{implementation}, a description of the tool that executes Rust programs under the model and generates the visualizations.

For a Rust learner, our models must be described in terms of Rust's surface syntax. Rust learners do not see or think in terms of core calculi or intermediate representations. However, this fact is in tension with our own need to design models that are both (1) simple to formally reason about and (2) feasible to implement. To resolve this tension, the formal model and implementation are described using the ``Mid-level Intermediate Representation,'' or MIR\,\citeyearpar{mir}, a control-flow graph IR within the Rust compiler. The informal model is described using the surface syntax, and the  implementation uses source-mapping information to lift the analysis from MIR to the surface.

%

\subsection{Dynamic Model}
\label{sec:dynsem}

\begin{figure}
    \centering
    \includegraphics[width=0.9\linewidth]{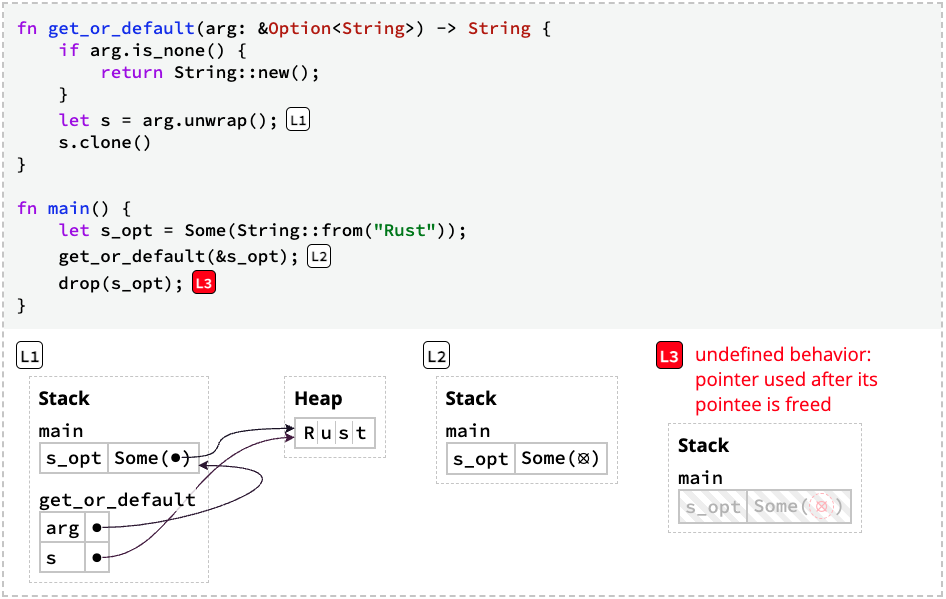}
    \cprotect\caption{The dynamic model visualized for the \rs{get_or_default} Ownership Inventory program. The state of memory is shown at three locations L1-L3. At L1, the diagram includes a heap pointer to the string data ``Rust'' and a stack pointer from the callee to the caller. At L2, the string has been deallocated on behalf of \rs{s} after the call to \rs{get_or_default} ends. At L3, undefined behavior occurs upon a double-free of \rs{s_opt}. Note that the L-$n$ labels and the \ub{} text are part of the visualization, not edited into the figure.}
    \label{fig:runtime-vis}
\end{figure}

An essential property of the dynamic model is that it must be able to express the \ub{} in Rust programs that is normally caught by the borrow checker. Conveniently, a similar need already exists in the Rust ecosystem to find \ub{} caused by \rs{unsafe} blocks. Miri\,\citeyearpar{miri} is a MIR interpreter that instruments a program's runtime to detect \ub{} like out-of-bounds memory accesses and use-after-free, comparable to Valgrind. Miri  provides a \textit{de facto} dynamic semantics to the MIR that can express \ub{} while avoiding unnecessary details like compiler optimizations.

Therefore, our dynamic model is basically ``what Miri does.'' The MIR does not have a formal semantics, although multiple projects are currently underway to design one\,\cite{minirust,amirformality}. Rather than provide a complete formal semantics for MIR, we will instead provide a didactic subset of Miri's semantics that suffices to explain our pedagogy.

\subsubsection{Informal Model}

A Rust program acts upon memory organized into a stack of function frames and a heap of long-lived data. Each stack frame maps syntactically-scoped variables to values, which are primitives (ints, bools, etc.), composites (structs, enums), or pointers. A path describes a particular value in memory, and pointers are essentially ``paths as values'' (as opposed to numeric addresses that can be arithmetically manipulated).

\Cref{fig:runtime-vis} shows an example of how we visualize the dynamic model. The state of memory is visualized at multiple points throughout the program. The diagram is fairly similar to prior work in program state visualization\,\cite{sorva2013}, so we will not belabor its design. In brief: one key aspect is that states are unrolled over time. Unrolling lets a person more easily compare changes between states, and it allows a person to see all the information in the diagram without interaction. This contrasts with tools like Python Tutor\,\cite{guo2013}, which shows one state at a time and requires users to actively scrub a slider between states. Another key aspect is that certain core data-types like boxes, strings, and vectors can be abstracted. For example, the abstracted version  of \rs{s_opt} just shows a pointer to heap data, while the expanded version would show a struct containing fields like the string's length (a button in the diagram permits toggling between these views).



\subsubsection{Formal Model}
\label{sec:dynformalmodel}

\setlength{\fboxrule}{0.5pt}
\newcommand{\synhead}[2]{%
    \begin{flushleft}\hspace{#2}\fbox{\textsc{#1}}\end{flushleft}}

\begin{figure}
\synhead{Programs}{0em}
\vspace{-3em}

\begin{align*}
\msf{Variable}~x \hspace{12pt} \msf{Number}~n \hspace{12pt} \msf{Function~Name}~f
\end{align*}

\begin{minipage}{0.43\textwidth}
$$
\begin{aligned}
\bnfm{Projection}{q}{\varepsilon \mid q.n} \\
\bnfm{Path}{p}{x \mid p.q \mid \ast p} \\
\bnfm{Constant}{c}{n \mid \msfb{true} \mid \msfb{false}} \\
\bnfm{Ownership~Qualifier}{\omega}{\shrd \mid \uniq} \\
\bnfm{Loan}{\ell}{\borrow{p}{\omega}}
\end{aligned}
$$
\end{minipage}
\hfill
\begin{minipage}{0.53\textwidth}
$$
\begin{aligned}
\bnfm{Rvalue}{rv}{c \mid p \mid \ell \mid (\overline{p_i})} \mid \msfb{box}~p \\
\bnfm{Instruction}{I}{}
    p := rv \mid \msfb{if}~p~\msfb{then}~n_1~\msfb{else}~n_2 \mid f(\overline{p_i}) \mid \\ 
    &\msfb{return}~p \mid \msfb{drop}~p \\
\bnfm{CFG}{G}{\seqof{I_i}} \\
\bnfm{Function}{F}{\msfb{fn}~f\seqof{\varrho_1 :> \varrho_2}(\overline{x_i : \tau_i}) \rightarrow \tau_o ~ \{ G \}}    
\end{aligned}
$$
\end{minipage}

\vspace{1em}

\begin{minipage}[t]{0.45\textwidth}
\synhead{Types}{0em}
\begin{gather*}
\begin{aligned}
\msf{Concrete~Lifetime}~r \hspace{12pt} \msf{Abstract~Lifetime}~\varrho
\end{aligned}
\\
\begin{aligned}
    \bnfm{Lifetime}{l}{r \mid \varrho} \\
    \bnfm{Type}{\tau}{}
        \msfb{u32} \mid \msfb{bool} \mid (\tau_1, \ldots, \tau_n) \mid \\ 
        & \&l~\omega~\tau \mid \msfb{box}~\tau
\end{aligned}
\end{gather*}
\end{minipage}
\hfill
\begin{minipage}[t]{0.49\textwidth}
\synhead{Runtime}{0em}
\vspace{-3em}

\begin{gather*}
\begin{aligned}
\msf{Heap~Location}~\heaploc
\end{aligned}
\\
\begin{aligned}
\bnfm{Segment}{s}{\msfb{frame}~(n, x) \mid \msfb{heap}~\heaploc} \\
\bnfm{Address}{a}{s.q} \\
\bnfm{Value}{v}{c \mid a \mid (v_1, \ldots, v_n)} \\
\bnfm{Env}{\env}{\seqof{x_i \mapsto v_i}} \\
\bnfm{Frame}{\sigma}{(f, n, \env)} \\
\bnfm{Stack}{S}{\seqof{\sigma_i}} \\
\bnfm{Heap}{H}{\seqof{\heaploc_i \mapsto v_i}}
\end{aligned}
\end{gather*}
\end{minipage}

\caption{The syntax and runtime structure of MIR Rust programs.}
\label{fig:syntax}
\end{figure}

\Cref{fig:syntax} provides the syntax for a subset of the MIR. A control-flow graph $G$ consists of a sequence of indexed instructions $I$. Instructions are either assignments, conditionals, function calls, returns, or deallocations. The basic primitives are standard (numbers, booleans, tuples, functions), but the interesting operations are those that involve memory.

Memory is arranged into two segments: a stack $S$ of frames $\sigma$, and a heap $H$ that maps locations $\heaploc$ to values $v$. Stack allocations are created as frame-local variables through instructions such as $x := (0, 1)$, and heap allocations are created with boxes such as $y := \msf{box}~2$. Data in memory are accessed via paths $p$, such as $x.0$ and $\ast y$. Finally, references to paths can be created with loans $\ell ::= \borrow{p}{\omega}$, where $\omega$ qualifies the loan as either shared ($\shrd$) or unique ($\uniq$). (We will say more about the type system when discussing the static model.)

Miri's operational semantics for this model is a judgment $\steps{(S, H)}{(S', H')}$. The key idea is that Miri's semantics can express \ub{} such as dangling pointers. Consider this program:
$$
x := \msf{box}~0;~ \msf{drop}~x;~ y := \ast x
$$

\noindent Miri will give the following execution trace for this program: 
%
\begin{alignat*}{4}
&(S = (x := \msfb{box}~0,~ &&\varnothing &),~ &H = \varnothing &&) \stepop \\
&(S = (\msfb{drop}~x,~ &&\{x \mapsto \heaploc\} &),~ &H = \{\heaploc \mapsto 0\} &&) \stepop \\
&(S = (y := \ast x,~ &&\{x \mapsto \heaploc\} &),~ &H = \varnothing &&) \stepop \\
&\text{\ub}
\end{alignat*}

Miri's semantics are not expressive enough to represent aspects of the compiler like register allocation and auto-vectorization. But if those optimizations are implemented correctly, a Rust program that does not have \ub{} under Miri's semantics should also not have \ub{} under the Rust compiler's actual semantics for its various assembly targets.

\subsubsection{Implementation}
\label{sec:dynimpl}

Given a Rust program, we run Miri to collect an execution trace that describes the state of the stack and heap after each instruction. We use the debug information within Miri (e.g., the types of variables on the stack) to reconstruct the structure of the data in memory. Based on source-mapping information in the Rust compiler, we group contiguous subsequences of the trace into steps that represent source-level expressions and statements. Then this data is passed to a script in a web browser which visualizes the data through a combination of HTML, CSS, and Javascript.

One subtlety is that Miri does not normally execute programs which would be rejected by the borrow checker, such as the one in \Cref{fig:runtime-vis}. However, in the current implementations of Rust and Miri, the borrow checker does not substantially influence the translation of a program into an Miri-interpretable CFG. Therefore, we can configure the compiler to simply ignore borrow checking errors, and then ask Miri to attempt to execute programs regardless. Miri is already well-suited to catching \ub{} (just normally in \rs{unsafe} code), so it suffices for catching \ub{} in unchecked code in the safe subset of Rust as well.

\subsection{Static Model}
\label{sec:statsem}

The dynamic model provides the foundation for understanding how programs can go wrong. The static model should then help learners understand how Rust's borrow checker catches programs that could go wrong (soundness), as well as when safe programs may be rejected (incompleteness).

As a flow-sensitive analysis, borrow checking is more complex than the usual type system encountered by today's programmers. So in designing a conceptual model of the borrow checker, our main goal was to condense the complex intermediate state of the analysis into a comprehensible, visualizable object. The result is the \emph{permissions model} of borrow checking.

\subsubsection{Informal Model}

\begin{figure}
    \hfill
    \begin{subfigure}{0.47\linewidth}
        \centering
        \includegraphics[width=0.9\linewidth]{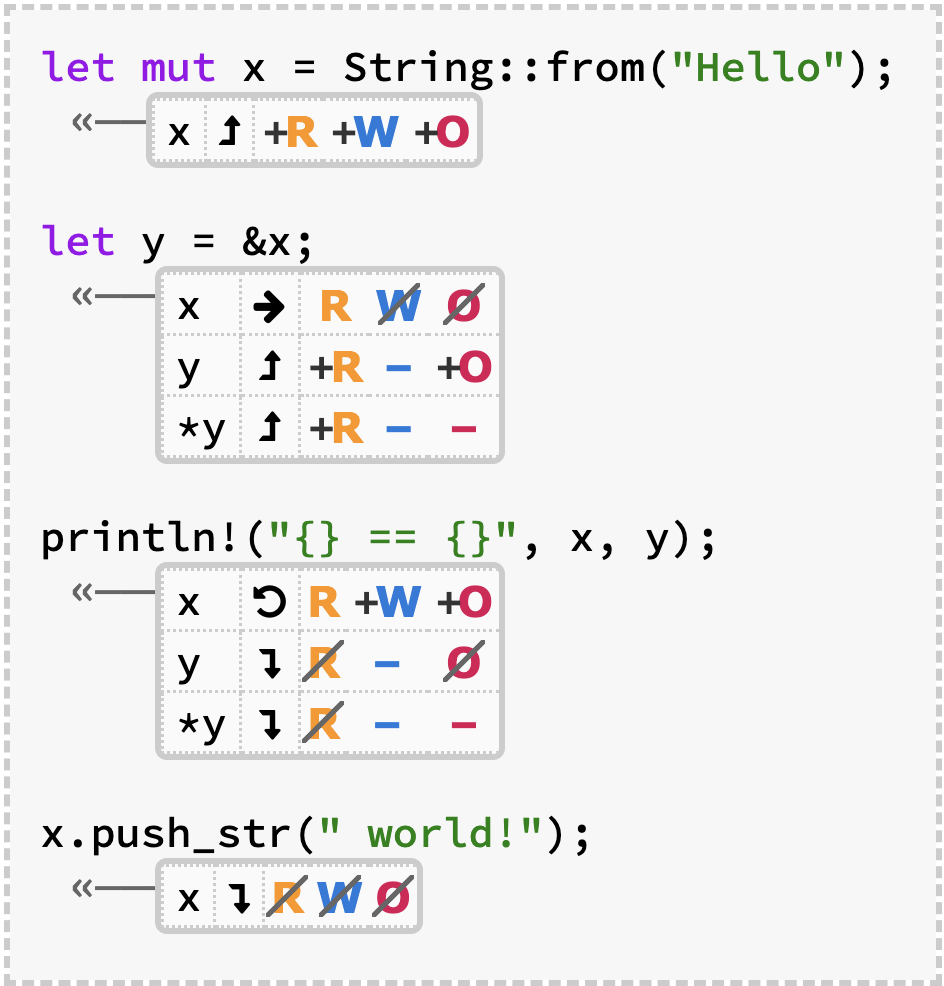}
        \caption{Each table shows the changes in permission state after a given statement.}
        \label{fig:static-diagrams-left}
    \end{subfigure}    
    \hfill
    \begin{subfigure}{0.49\linewidth}
        \centering
        \includegraphics[width=0.9\linewidth]{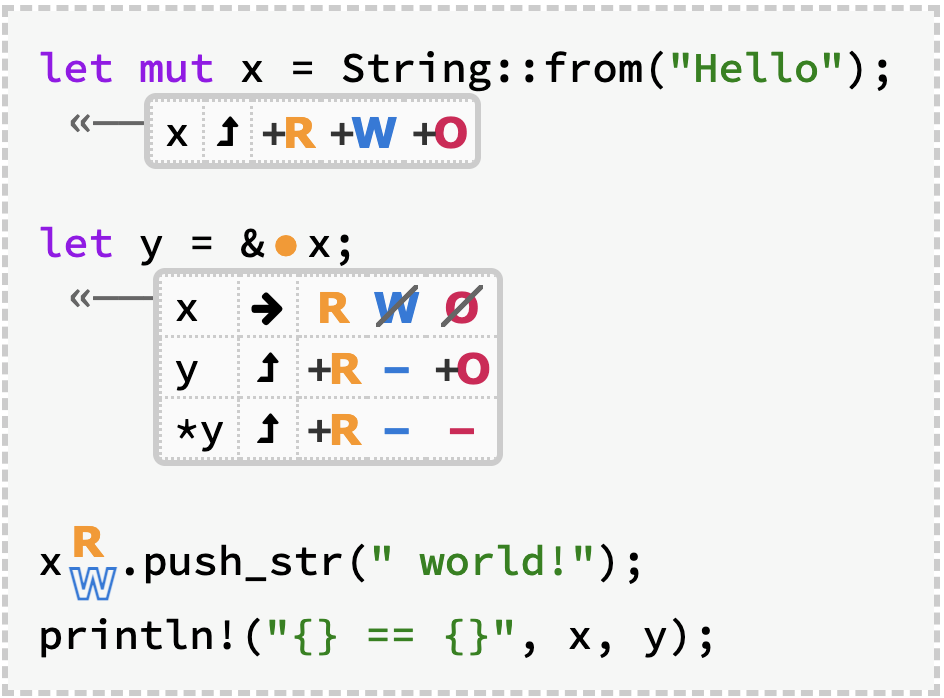}
        \caption{Operations on paths indicate which permissions are expected, and whether those permissions exist.}
        \label{fig:static-diagrams-right}
    \end{subfigure}
    \hfill
    \vspace{-0.5em}
    \cprotect\caption{Visualizations of the permissions model over two programs that borrow a string.}
    \label{fig:static-diagrams}
\end{figure}

At compile-time, Rust's borrow checker uses a system of \emph{permissions} to check whether an operation might cause \ub{}. The borrow checker tracks whether a path is readable ($\rpm$), writable ($\wpm$), or ownable ($\opm$). Focusing on \Cref{fig:static-diagrams-left}, a variable has $\rpm\opm$ permissions by default, and it has the $\wpm$ permission if declared with \rs{let mut}. The string \rs{x} therefore has $\rpm\wpm\opm$ permissions. With these permissions, the string can be read like \rs{x.len()}, written like \rs{x.push_str(..)}, and owned like \rs{drop(x)}. The plus sign indicates that the permissions were gained, and the cause of the change is indicated by the icon (up-arrow \birthIco\space for variable initialization). The borrow \rs{let y = &x} removes $\wpm\opm$ permissions from \rs{x} (right-arrow \refineBeginIco\space for path borrowed); this action provides $\rpm\opm$ permissions to \rs{y} and the $\rpm$ permission to \rs{*y}. Once \rs{y} is no longer used, then its permissions are eliminated (down-arrow \deathIco\space for the end of a variable's live range) and \rs{x} regains $\wpm\opm$ permissions (cycle-left \refineEndIco\space for regaining borrowed permissions). 

Turning to \Cref{fig:static-diagrams-right}, operations expect permissions from paths. The  expectations are visually placed between the operator requiring permissions and the path operand. The borrowing operation \rs{&x} expects that \rs{x} has the $\rpm$ permission, represented as a yellow circle. The circle is filled in because \rs{x} has the $\rpm$ permission. The operation \rs{x.push_str(...)} expects $\rpm\wpm$ permissions, represented by the stack of the two letters. (The user of the diagram generator can place annotations on the source code to make an expectation be represented as either a letter or a circle, depending on where the user wants to focus their readers.) However, by swapping the order of the \rs{push_str} and \rs{println} lines, \rs{x} no longer has the $\wpm$ permission, visualized as a hollow letter. Therefore the borrow checker will reject this program. 
\subsubsection{Formal Model}
\label{sec:statformalmodel}

Rust's ownership types are often viewed as enforcing ``aliasing XOR mutability'', which recent work has started to articulate through the metaphor of permissions\,\cite{yanovski2021ghost}: what can or can't a program do on a particular path at a particular program location? We advance this idea by designing a conceptual model of the borrow checker that reifies permissions into formal objects that can be analyzed, visualized, and taught to Rust learners.

To precisely characterize the permissions model, we first need to provide a model of how Rust's borrow checker actually works. Polonius\,\cite{polonius} is a model of the borrow checker that is maintained by the Rust compiler developers and implemented in Datalog. The former aspect means that Polonius is likely to be a plausible model of Rust's implementation. The latter aspect enables us to easily implement our own alternative model that shares a base of facts about properties like liveness. Sharing facts simplifies both our implementation and our proof of model equivalence.


\begin{figure}
    \centering
\begin{mathpar}
\ir{Pol-Borrow-Conflict}
    {\liveat{\ell}{I} \\ \invalidat{\ell}{I}}
    {\accesserror{G}}
    {b-borrow-conflict}

\ir{Pol-Move-Conflict}
    {\movedbefore{p}{I} \\ \readat{p}{I}}
    {\accesserror{G}}
    {b-move-conflict}

\\

\ir{Pol-Read-Invalid}
    {\readat{p'}{I} \\ p \conflicts p'}
    {\invalidat{\borrow{p}{\uniq}}{I}}
    {b-read-invalid}

\ir{Pol-Write-Invalid}
    {\writtenat{p'}{I} \\ p \conflicts p'}
    {\invalidat{\borrow{p}{\omega}}{I}}
    {b-write-invalid}

\ir{Pol-Move-Invalid}
    {\movedat{p'}{I} \\ p \conflicts p'}
    {\invalidat{\borrow{p}{\omega}}{I}}
    {b-move-invalid}
\end{mathpar}
    \vspace{-0.5em}
    \caption{The core subset of inferences rules for the Polonius model of the borrow checker.}
    \label{fig:polonius-rules}
\end{figure}

\paragraph{\hspace{-\parindent}Polonius model of borrow checking.}
\Cref{fig:polonius-rules} shows a subset of the inference rules for the Polonius model. Both the Polonius and permission models rely on a shared set of judgments about aspects like liveness and mutation. We do not define these judgments back to their axioms within this paper --- we refer interested readers to the Polonius source code\,\citeyearpar{polonius-src}. Instead, we just provide enough context to understand the differences and equivalences of the two model.

Given a control-flow graph $G$, Polonius will reject $G$ (written as ``$\accesserror{G}$'', consistent with Polonius' naming conventions)\footnotemark{} under one of two conditions (note that $I \in G$ for all rules):

\footnotetext{Type systems are normally formalized as a ``positive'' judgment, e.g., a program type-checks by constructing  a proof of $\Gamma \vdash e : \tau$. But Polonius is formulated as a ``negative'' judgment: a program does not type-check if a proof of $\accesserror{G}$ is constructed. For consistency with Polonius, we follow the negative judgment convention in our own model.}

\begin{description}

\item[\cref{tr:b-borrow-conflict}:] ``$\liveat{\ell}{I}$'' means that a loan $\ell ::= \borrow{p}{\omega}$ was created somewhere and is live at an instruction $I$, i.e., $\ell$ is used at $I$ or at some instruction reachable from $I$. ``$\invalidat{\ell}{I}$'' means $I$ performs an operation that conflicts with $\ell$. The $\msfb{invalidated\,at}$ judgment is defined in three ways. \cref{tr:b-read-invalid} states that a unique loan on a path $p$ is invalidated by a read of a conflicting path $p'$ (written $p \conflicts p'$). \cref{tr:b-write-invalid} states that any loan on a path $p$ is invalidated by a conflicting write. \cref{tr:b-move-invalid} states that any loan on a path $p$ is also invalidated by a conflicting move.

\begin{minipage}{0.79\textwidth}
\vspace{0.2em}\hspace{10pt} 
For example, consider the program on the right. Because $y$ is used in $z := \ast y$, then the loan $\borrow{x}{\shrd}$ is live at the instruction $x.0 := 1$. However, $x.0 := 1$ invalidates that loan because $x.0 \conflicts x$, and borrowed data cannot be mutated. Therefore this program has a loan conflict and is rejected.
\end{minipage}
\begin{minipage}{0.2\textwidth}
\vspace{-1em}
\begin{align*}
&x := (0, 0); \\[-0.3em]
&y := \borrow{x}{\shrd}; \\[-0.3em]
&x.0 := 1; \\[-0.3em]
&z := \ast y;
\end{align*}
\end{minipage}
\vspace{0.3em}

\item[\cref{tr:b-move-conflict}:] ``$\movedbefore{p}{I}$'' means that the path $p$ has been moved before reaching $I$. Any use of a movable data type (like a box) that is not through a reference will cause a move. ``$\readat{p}{I}$'' means that $p$ is read at $I$. 

\begin{minipage}{0.79\textwidth}
\vspace{0.2em}\hspace{10pt} 
For example, consider the program on the right. Because $x$ is moved by $y := x$, and $x$ is read later at $\ast x$, then this program has a move conflict and is rejected by the borrow checker.\end{minipage}
\begin{minipage}{0.2\textwidth}
\vspace{-1em}
\begin{align*}
&x := \ebox{0}; \\[-0.3em]
&y := x; \\[-0.3em]
&z := \ast x
\end{align*}
\end{minipage}
\vspace{0.3em}
\end{description}

\paragraph{\hspace{-\parindent}Permissions model of borrow checking.}
Next, we describe our permissions model and its relationship to the borrow checker. The basic idea is that the five different judgments used within the two Polonius $\msfb{access\text{-}error}$ rules can be abstracted into two higher-level judgments: ``$\needsat{p}{c}{I}$'' and ``$\missingat{p}{c}{I}$'', where $c ::= \rpm \mid \wpm \mid \opm$ is a permission to read, write, or own a path, respectively. A program has a permission violation (written as ``$\permfail{G}$'') under the \cref{tr:p-fail} rule, where a permission $c$ on a path $p$ is needed but missing at an instruction $I$.

\Cref{fig:permission-rules} shows the rules for the permissions model. 
The $\msfb{needs\,at}$ rules are straightforward: a path $p$ needs the $\rpm$ permission if read, the $\wpm$ permission if written, and the $\opm$ permission if moved. The $\msfb{missing\,at}$ rules describe the conditions under which a path lacks a particular permission. The \cref{tr:p-missing-read} rule states that a place $p$ cannot be read while a loan to a conflicting place $p'$ is live. This is analogous to \cref{tr:b-read-invalid}, i.e. that a read of place $p'$ invalidates any loan of a conflicting place $p$. The analogy can be formalized as a correctness theorem: we want to show that $\msfb{permission\text{-}error}$ soundly approximates $\msfb{access\text{-}error}$, i.e., that $\accesserror{G} \vdash \permfail{G}$. A simple rearrangement of terms proves this entailment for the case of \cref{tr:b-read-invalid}:
\begin{mathpar}
\infer
    {\liveat{\borrow{p'}{\uniq}}{I} \\ 
     \infer*[leftskip=0.5em]{\readat{p}{I} \\ p \conflicts p'}{\invalidat{\borrow{p'}{\uniq}}{I}}}
    {\accesserror{G}}
\vdash    
\infer
    {\infer*[rightskip=0.5em]
        {\readat{p}{I}}
        {\needsat{p}{\rpm}{I}} \\ 
     \infer*
        {\liveat{\borrow{p'}{\uniq}}{I} \\ 
         p \conflicts p'}
        {\missingat{p}{\rpm}{I}}}
    {\permfail{G}}
\end{mathpar}

\begin{figure}
    \centering
\begin{mathpar}
\ir{Perm-Fail}
    {\needsat{p}{c}{I} \\ \missingat{p}{c}{I}}
    {\permfail{G}}
    {p-fail} 

\ir{Perm-Needs-R}
    {\readat{p}{I}}
    {\needsat{p}{\rpm}{I}}
    {p-needs-read}

\ir{Perm-Needs-W}
    {\writtenat{p}{I}}
    {\needsat{p}{\wpm}{I}}
    {p-needs-write}

\ir{Perm-Needs-O}
    {\movedat{p}{I}}
    {\needsat{p}{\opm}{I}}
    {p-needs-own}

\ir{Perm-Missing-R}
    {\liveat{\borrow{p'}{\uniq}}{I} \\ p \conflicts p'}
    {\missingat{p}{\rpm}{I}}
    {p-missing-read}

\ir{Perm-Missing-WO}
    {\liveat{\borrow{p'}{\omega}}{I} \\ p \conflicts p'}
    {\missingat{p}{c \in \{\wpm, \opm\}}{I}}
    {p-missing-write-own}

\ir{Perm-Missing-$\ast$}
    {\movedbefore{p}{I}}
    {\missingat{p}{c \in \{\rpm, \wpm, \opm\}}{I}}
    {p-moved-no-permissions}
\end{mathpar}

    \caption{The inference rules for the permissions model of borrow checking.}
    \label{fig:permission-rules}
    \vspace{-0.5em}
\end{figure}

\noindent We can see a similar correspondence for the remaining rules. \cref{tr:p-missing-write-own} states that a path $p$ is missing the write and own permissions while there exists a live loan to a conflicting path $p'$. Correspondingly, \cref{tr:b-write-invalid} states that a loan on $p'$ is invalidated by a write to $p$. \cref{tr:p-moved-no-permissions} states that a path $p$ is missing all permissions after being moved. Correspondingly, \cref{tr:b-move-conflict} states that a use of a $p$ cannot occur if $p$ is moved. These correspondences are formalized in each case of the proof in \Cref{sec:proofs}.

\paragraph{\hspace{-\parindent}An open problem: lifetime parameter errors} 
Access errors are the most common kind of borrow checker error encountered by Rust users, accounting for 6/8 programs in the Ownership Inventory. The other kind of error is a \emph{lifetime parameter error}.

\noindent\begin{minipage}{0.66\textwidth}
\vspace{0.2em}\hspace{10pt} 
Consider the identity function on the right. Because $x$ flows into the return value, the lifetime $\varrho_1$ must ``outlive'' the lifetime $\varrho_2$. Polonius uses this information to modularly analyze the live ranges of references across function calls, e.g. $b := \msf{id}(a)$ should cause $a$ to be live as long as $b$ is live. However, Rust does not infer outlives-constraints for function types, so the user must explicitly specify $\varrho_1 :> \varrho_2$ or else this function would be rejected.

\end{minipage}
\hfill
\begin{minipage}{0.25\textwidth}
\vspace{-1em}
\begin{align*}
& \msfb{fn}~\msf{id}\langle \varrho_1, \varrho_2\rangle(x{: \&\varrho_1~\uniq~\msf{u32}}) \\[-0.3em]
& \hspace{8pt}\rightarrow \&\varrho_2~\uniq~\msf{u32} \\[-0.3em] 
& \{ \\[-0.3em]
& \hspace{12pt}y := x; \\[-0.3em]
& \hspace{12pt}\msfb{return}~y \\[-0.3em]
& \}
\end{align*}
\end{minipage}

\vspace{0.1em}
Permissions are not a perfect analogy to explain lifetime parameter errors. For example, an outlives-constraint cannot always be blamed on a path, while our model is structured around path-specific permissions. We have experimented with a fourth kind of ``flow'' permission $\fpm$:
\begin{mathpar}
\ir{Pol-Lifetime-Conflict}
    {\outlivesat{\varrho_1}{\varrho_2}{I} \\ \varrho_1 \notoutlives \varrho_2}
    {\subseterror{G}}
    {b-lifetime-conflict}

\hspace{-2em}

\ir{Perm-Needs-Flow}
    {\flowsto{\varrho_1}{\varrho_2}{p}{I}}
    {\needsat{p}{\fpm}{I}}
    {p-needs-flow}

\hspace{-2em}    
    
\ir{Perm-Missing-Flow}
    {\flowsto{\varrho_1}{\varrho_2}{p}{I} \\ \varrho_1 \notoutlives \varrho_2}
    {\missingat{p}{\fpm}{I}}
    {p-lifetime-missing-flow}
\end{mathpar}

The first rule states that Polonius finds a lifetime parameter error when an instruction $I$ requires $\varrho_1$ to outlive $\varrho_2$, but the function is not annotated with that outlives-constraint. The corresponding permission rules narrow the scope to outlives-constraints that can be blamed on a place $p$ (such as $y$ in the $\msf{id}$ example). The rules state that $p$ needs $\fpm$ if such a constraint exists, and that $p$ is missing $\fpm$ if the function lacks the necessary outlives annotation. Due to the narrowing of scope, our model does not soundly approximate lifetime parameter errors. 
In future work, we will investigate either extensions to the permission model or alternative conceptual models that can better represent lifetime parameter errors.

\subsubsection{Implementation}
\label{sec:statimpl}

\begin{figure}
    \begin{subfigure}[t]{0.48\linewidth}
        \centering
        \includegraphics[width=\textwidth]{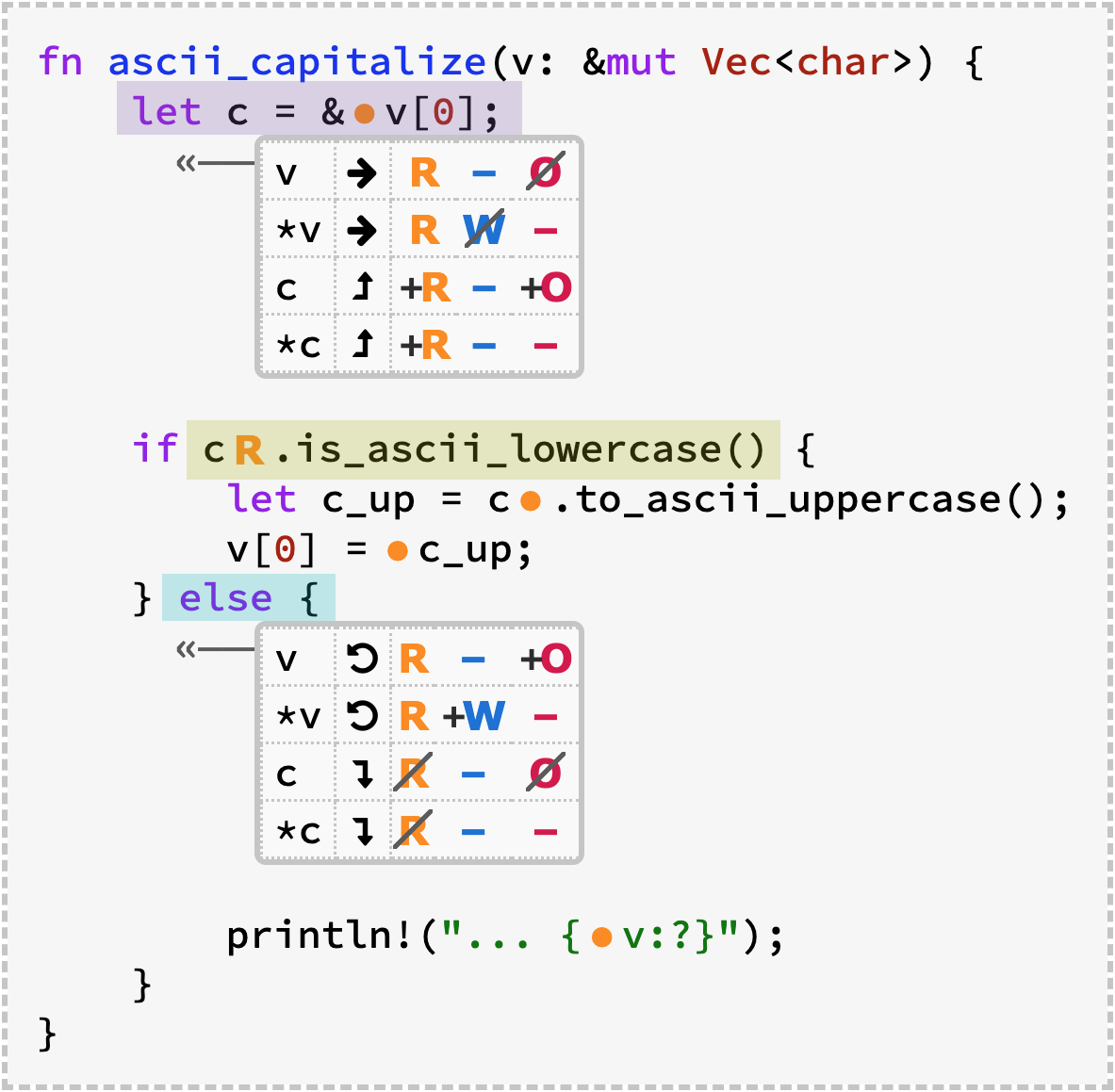}
        \cprotect\caption{The function \rs{ascii_capitalize} capitalizes the first character in a vector of ASCII characters. It demonstrates flow-sensitive changes in permissions.}
        \label{fig:static-impl-left}
    \end{subfigure}
    \hfill
    \begin{subfigure}[t]{0.48\linewidth}
        \centering
        \includegraphics[width=0.66\textwidth]{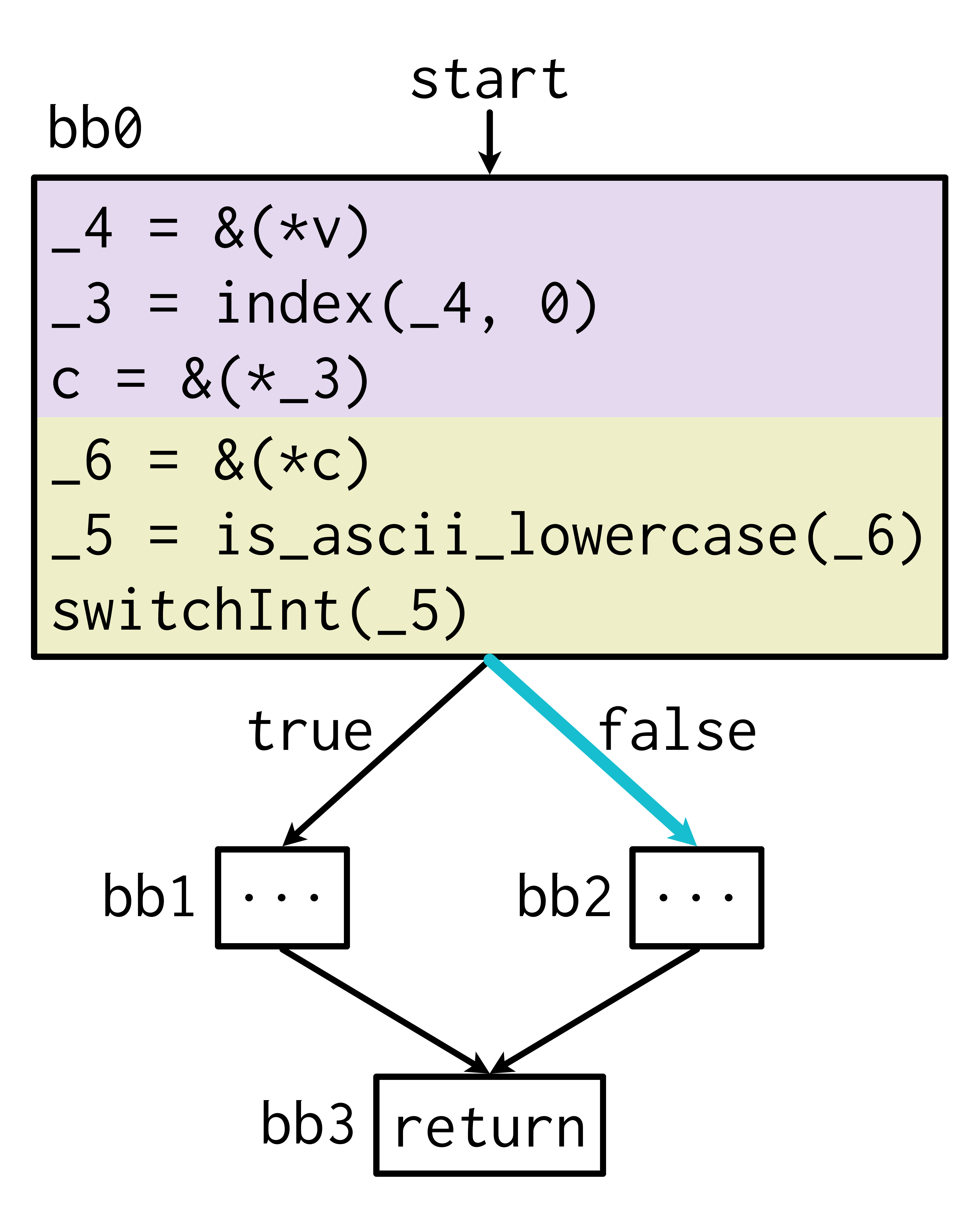}
        \cprotect\caption{A simplified MIR control-flow graph for \rs{ascii_capitalize} highlighting relevant parts of the CFG and how they map to the source-level.  }
        \label{fig:static-impl-right}
    \end{subfigure}    
    
    \caption{Example of how the MIR CFG relates to source-code constructs.}
    \label{fig:static-impl}
\end{figure}

The structure of our static model visualization parallels the \cref{tr:p-fail} rule: one component for the $\msfb{needs\,at}$ relation and one component for the $\msf{missing\,at}$ relation, corresponding to the letters and tables shown in \Cref{fig:static-diagrams}.
After generating these relations at the MIR level following \Cref{fig:permission-rules}, the main implementation challenge is to lift the relations to the source-level.
We will briefly discuss why this source-mapping is non-trivial and how we approach it, using the \rs{ascii_capitalize} function in \Cref{fig:static-impl} as a running example.

\paragraph{\hspace{-\parindent}Needs-at analysis.} 
One challenge for the $\msfb{needs\,at}$ analysis arises with desugared conversions. For example, consider the method call \rs{c.is_ascii_lowercase()} in \Cref{fig:static-impl-left}. Intuitively, the method's type signature \rss{&char -> bool} means that the $\rpm$ permission should be needed from the method's receiver, the path \rs{c}.
However, as shown in \rs{bb0} of \Cref{fig:static-impl-right}, the MIR-level method is not called directly on \rs{c}, but rather an automatically generated temporary \rs{_6} that is a reborrow of \rs{c}. Therefore, to lift the needs-at analysis for method calls, we have to search backwards from the MIR call instruction to find the first use of the source-visible receiver path.

 
\paragraph{\hspace{-\parindent}Missing-at analysis.}
The $\msf{missing\,at}$ relation defines a permission state, or a location-specific set of missing permissions for each path. Rather than visualize the entire permission state at each instruction, we instead visualize the differences in permission state (or ``steps'') caused by each instruction, which help readers better see how operations affect permissions. It is straightforward to compute steps between adjacent MIR instructions, but the implementation challenge is to determine which clusters of MIR instructions correspond to meaningful source-level steps.

For example, at the beginning of \rs{ascii_capitalize} in \Cref{fig:static-impl}-left, \rs{c} is defined as a shared borrow of \rs{v[0]}. As a result, \rs{v} loses $\opm$, \rs{*v} loses $\wpm$, \rs{c} gains $\rpm\opm$, and \rs{*c} gains $\rpm$. Using compiler source map information, we compute the contiguous subsequence of MIR instructions that correspond to the source-level statement; in \Cref{fig:static-impl} these are highlighted in \softHighlight[implPurple]{purple}. Then we compute the step as the difference in permission state between the first and last instructions.

Steps are not always intra-basic-block subsequences --- they can also span across basic blocks, as shown in \Cref{fig:static-impl}-right when the \rss{else} branch is taken. In this case the sequence doesn't even include any instructions, just the branch between blocks highlighted in \softHighlight[implCyan]{cyan}. These changes occur due to the flow-sensitive liveness of \rs{c}, which is indicated in our diagram by the down-arrow next to the permission changes. Using these techniques, we lift the formal model to a visual description displayed on the source language.

\section{A Pedagogy for Ownership}
\label{sec:pedagogy}

\begin{figure}
\begin{framed}
\begin{small}
\centering
\hfill
\begin{minipage}[t]{0.51\textwidth}
\begin{enumerate}
    \item[4.1] \textbf{What is Ownership?}
    \begin{enumerate}[label=4.1.\arabic*,leftmargin=1.5em]
        \raggedright
        \item Safety is the Absence of Undefined Behavior 
        \item Ownership is a Discipline for Memory Safety
        \item Variables Live in the Stack
        \item Boxes Live in the Heap
        \item Rust Does Not Permit Manual Memory Management
        \item A Box's Owner Manages Deallocation
        \item At Runtime, A Move is Just a Copy
        \item Collections Use Boxes
        \item Variables Cannot Be Used After Being Moved
        \item Cloning Avoids Moves
    \end{enumerate}
\end{enumerate}
\end{minipage}
\hfill
\begin{minipage}[t]{0.48\textwidth}
\begin{enumerate}
    \item[4.2] \textbf{References and Borrowing}
    \begin{enumerate}[label=4.2.\arabic*,leftmargin=1.5em]
        \raggedright
        \item References Are Non-Owning Pointers
        \item Dereferencing a Pointer Accesses Its Data
        \item Rust Avoids Simultaneous Aliasing and Mutation
        \item References Change Permissions on Paths
        \item The Borrow Checker Finds Permission Violations
        \item Mutable References Provide Unique and Non-Owning Access to Data
        \item Permissions Are Returned At The End of a Reference's Lifetime
        \item Data Must Outlives All Of Its References
    \end{enumerate}
\end{enumerate}
\vspace{0.5em}
\end{minipage}
\hfill
\begin{minipage}{0.8\textwidth}
\begin{enumerate}[resume*]
    \item[4.3] \textbf{Fixing Ownership Errors}
    \begin{enumerate}[label=4.3.\arabic*,leftmargin=1.5em]
        \item Fixing an Unsafe Program: Reference to the Stack
        \item Fixing an Unsafe Program: Not Enough Permissions
         \item Fixing an Unsafe Program: Aliasing and Mutating a Data Structure
        \item Fixing an Unsafe Program: Copying vs. Moving Out of a Collection
        \item Fixing a Safe Program: Mutating Different Tuple Fields
        \item Fixing a Safe Program: Mutating Different Array Elements
    \end{enumerate}
\end{enumerate}
\end{minipage}
\end{small}
\end{framed}

\caption{The table of contents for the three sections of our chapter on ownership, designed as a drop-in replacement for \trpl{} Chapter 4: ``Understanding Ownership.''}
\label{fig:pedagogy}

\end{figure}

\noindent The models in \Cref{sec:notionalmachines} provide the conceptual foundation for understanding the aspects of ownership identified in \Cref{sec:ci} such as \ub{} and incompleteness. Next, we designed a pedagogy that could help Rust learners internalize these models. Rather than designing an entire Rust curriculum from scratch, we instead  forked the popular open-source Rust textbook \textit{The Rust Programming Language} (\trpl{})\,\cite{trpl}. \trpl{} covers most of the language's core features, and it is the official Rust learning resource endorsed by the Rust project.


We designed a new pedagogy of ownership as a replacement for the existing chapter on ownership in \trpl{}. The structure of the pedagogy is apparent in the sequence of headings used to organize each section, shown in \Cref{fig:pedagogy}. We start by explaining the core ideas of undefined behavior and memory safety through boxes and moves (the dynamic model). We then introduce references, the borrow checker, and permissions (the static model). Finally, we synthesize these ideas by providing multiple examples of how a Rust programmer can interpret and fix ownership errors, emphasizing the distinction between soundness and completeness. The full text of the chapter is available online at this link: \url{https://rust-book.cs.brown.edu/}

\subsection{An Illustrative Excerpt}

We will provide a sense of the pedagogic principles we used in writing the chapter by walking through the design of the book's \S 4.3.4: ``Fixing an Unsafe Program: Copying vs. Moving Out of a Collection''. Each excerpt is boxed in gray, and the pedagogic justification is beside the box. 

The goal of this section is to help learners understand the distinction between movable types (like \rs{String} or \rs{Vec}) and copyable types (like \rs{i32} or \rs{bool}). The section starts like this:

\vspace{0.5em}

\noindent\begin{minipage}{0.6\textwidth}
\begin{excerpt}
    A common confusion for Rust learners happens when copying data out of a collection, like a vector. For example, here's a safe program that copies a number out of a vector:

    \vspace{0.5em}
    \begin{center}
    \noindent\includegraphics[width=\linewidth]{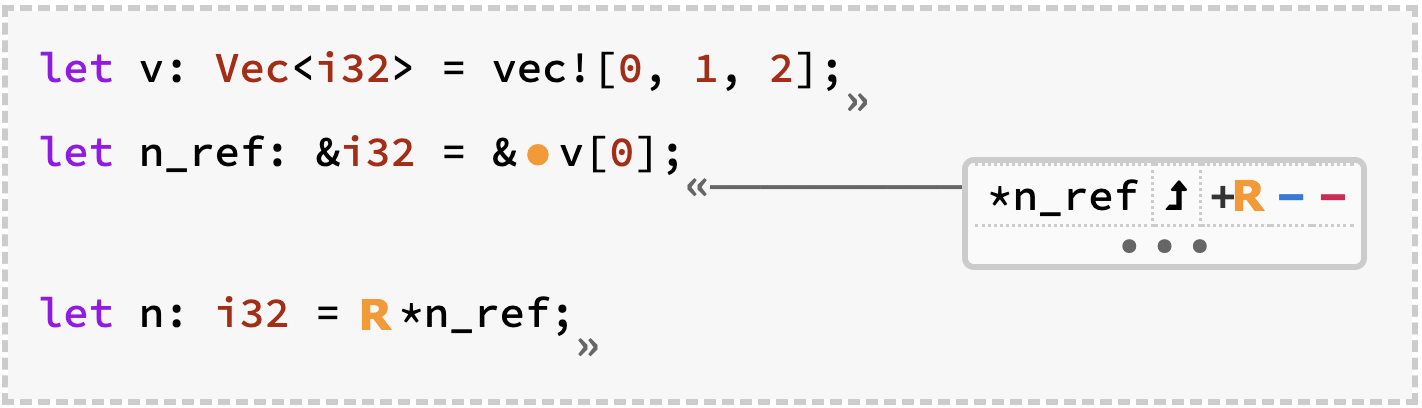}
    \end{center}
\end{excerpt}
\end{minipage}
\hfill
\begin{minipage}{0.39\textwidth}
    Each section is anchored around a concrete running example, like copying an element from a vector. Initially, the example is a valid, compiling program. The permission annotations show both the expected permissions on the relevant operation ($\rpm$ on \rs{*n_ref}) and how those permissions were gained (via the borrow \rs{&v[0]}).
\end{minipage}

\vspace{0.5em}

\noindent\begin{minipage}{0.35\textwidth}
    Then a small change is made to the program such that it no longer compiles (and in this case, is also now unsafe). The change is as small as possible so the reader can maximally transfer their understanding of the previous snippet onto the current one. The contrast in the permission diagrams emphasizes how the change in type has affected the internal state of the borrow checker.

\end{minipage}
\hfill
\begin{minipage}{0.63\textwidth}
\begin{excerpt}
    The dereference operation \rss{*n_ref} expects just the $\rpm$ permission, which the path \rss{*n_ref} has. But what happens if we change the type of elements in the vector from \rss{i32} to \rss{String}? Then it turns out we no longer have the necessary permissions:
    
    \vspace{0.5em}
    \begin{center}
    \includegraphics[width=\linewidth]{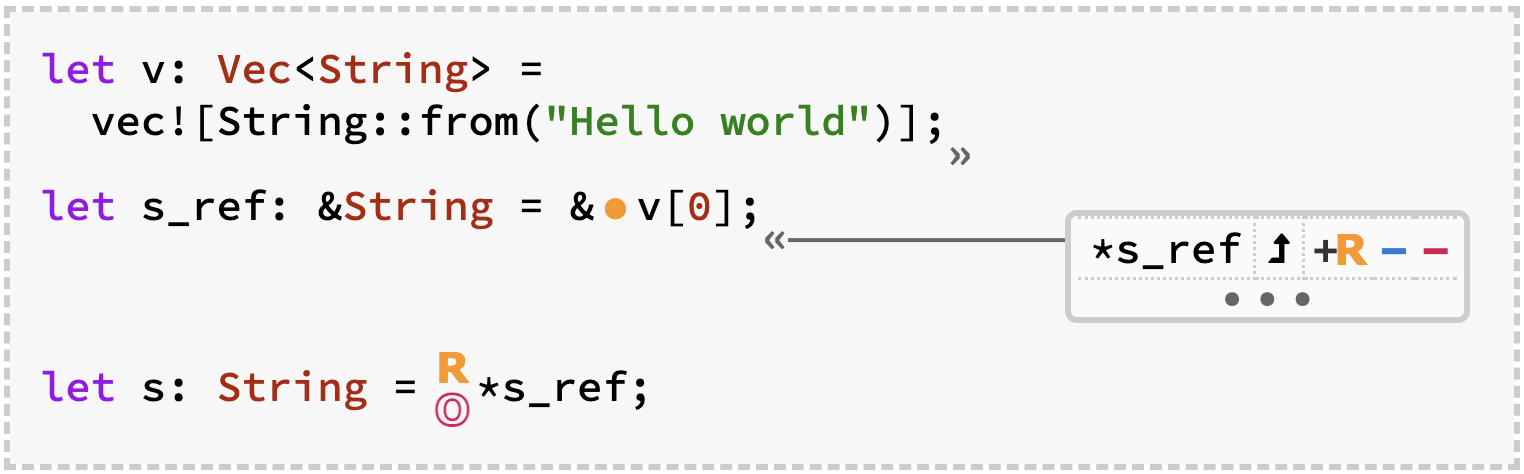}
    \end{center}
\end{excerpt}
\end{minipage}

\vspace{0.5em}

\noindent\begin{minipage}{0.7\textwidth}
\begin{excerpt}
    The first program will compile, but the second program will not compile. Rust gives the following error message:

    \vspace{0.5em}
    \begin{footnotesize}
    \begin{verbatim}
error[E0507]: cannot move out of `*s_ref` which is behind 
              a shared reference
 --> test.rs:4:9
  |
4 | let s = *s_ref;
  |         ^^^^^^
  |         |
  |         move occurs because `*s_ref` has type `String`, 
  |         which does not implement the `Copy` trait
    \end{verbatim}
    \end{footnotesize}

    \noindent The issue is that the vector \rss{v} owns the string ``Hello world''. When we dereference \rss{s_ref}, that tries to take ownership of the string from the vector. But references are non-owning pointers — we can't take ownership through a reference. Therefore Rust complains that we ``cannot move out of [...] a shared reference''.
\end{excerpt}
\end{minipage}
\hfill
\begin{minipage}{0.28\textwidth}
The reader is then given the actual output of the Rust compiler. These error messages will be the actual text encountered by Rust learners in their day-to-day practice, so it is important to explicitly relate the text of the error to the permissions model. (In future work we hope to incorporate the permissions visualizer into the IDE.)
\end{minipage}

\newpage

\begin{wrapfigure}{r}{0.62\textwidth}
\begin{excerpt}
    But why is this unsafe? We can illustrate the problem by simulating the rejected program:

    \vspace{0.5em}
    \begin{center}
    \includegraphics[width=\linewidth]{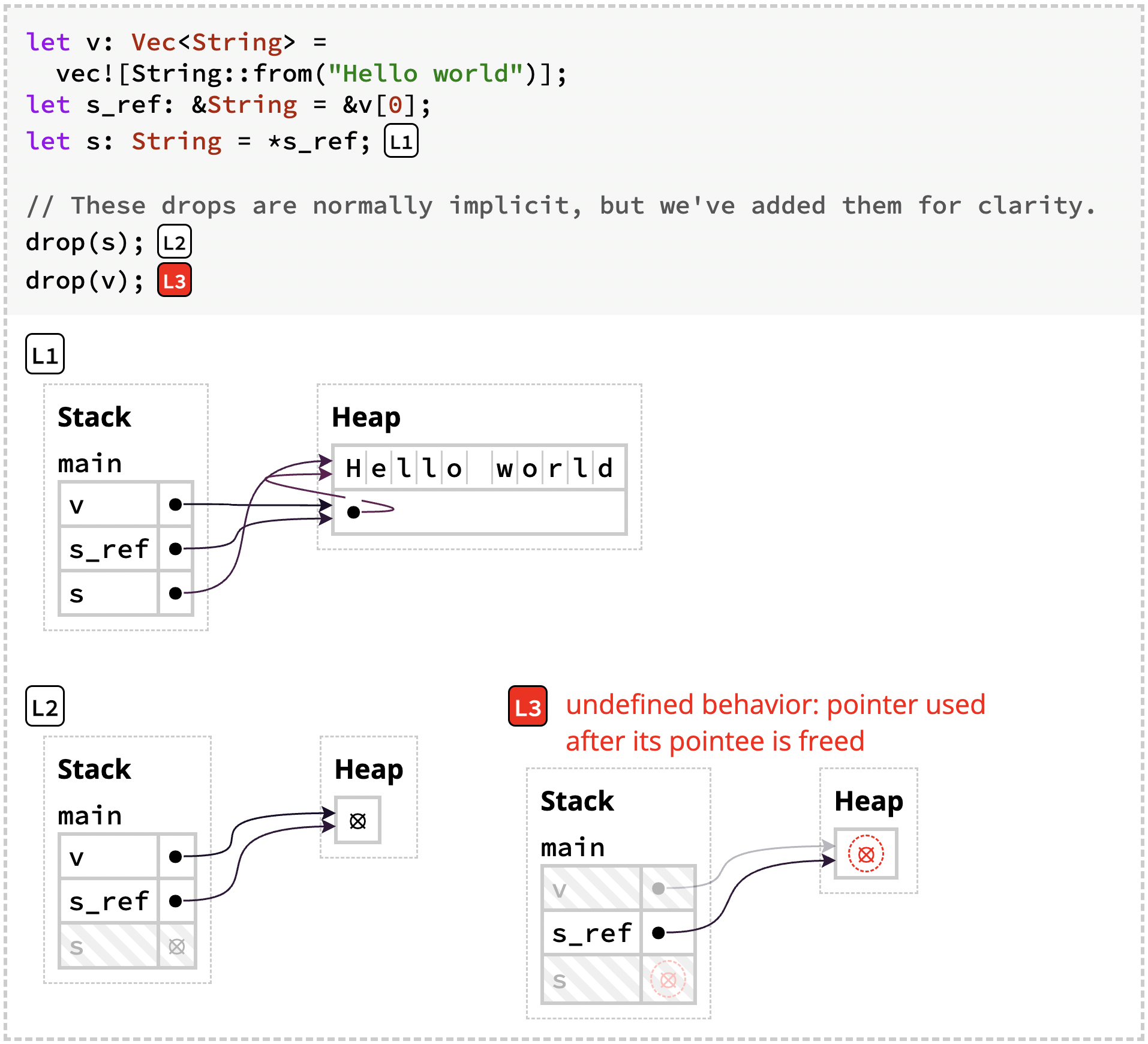}
    \end{center}

     \noindent What happens here is a \textbf{double-free}. After executing \rss{let s = *s_ref}, both \rss{v} and \rss{s} think they own ``Hello world''. After \rss{s} is dropped, ``Hello world'' is deallocated. Then \rss{v} is dropped, and undefined behavior happens when the string is freed a second time.
\end{excerpt}
\end{wrapfigure}

After establishing that a program is rejected by the compiler, we then engage in counterfactual reasoning: what would happen if the program were allowed to execute? In this instance, the code is already executable (i.e., not a abstract function), so we do not need to construct a separate counterexample. 

The steps in the diagram are purposefully selected: first, we show a reasonable initial state of memory with live pointers. Then, we show where a pointer becomes dangling. Finally, we show where the dangling pointer is used. 

In the remainder of the section (which we elide for brevity), the case study is generalized into a pithy principle: ``if a value does not own heap data, then it can be copied without a move''. Then the text explores a space of solutions to avoid the move, such as: only using a reference, deep copying the data, or consuming ownership by removing the element from the collection. 




We believe that this style of exposition provides readers with the requisite foundations to reason about errors like ``cannot move out of a shared reference'' from first principles. Notably, this is the same kind of error that stymied the StackOverflow questioner in \Cref{fig:so-questions}. In their SO post, that person wrote: ``I see there is already a lot of documentation about borrow checker issues, but after reading it, I still can't figure out the problem.'' We set out to determine whether our new pedagogy would suffice to help learners like this one in such cases.

\section{Evaluation}
\label{sec:eval}

We sought to evaluate our pedagogy on whether it helps learners understand ownership in Rust. This raised two immediate questions. First, how do we find learners to try out our pedagogy? The vast majority of CS education research takes place in a classroom, but we explored an alternative route: free online textbooks. These resources provide access to a larger and more diverse population of learners than CS undergraduates at a single institution. To that end, we set up a publicly-accessible website that hosts our \trpl{} fork, and it has been visited by tens of thousands of Rust learners to date. This site provides a research platform for analyzing and intervening in the Rust learning process. The intervention described in this paper is the first step in an ongoing experiment to leverage the platform for systematically improving Rust education at scale.

The second key question is: how do we know if learners understand ownership after following our pedagogy? ``Understand'' is difficult to define --- ideally, a longitudinal study might measure understanding as learners' ability to productively write safe and performant Rust code in their context of use. But for lack of such  data, we instead opted for a common substitute: quiz questions. For example, \citet{ongaro2014raft} faced a similar situation: evaluating a novel conceptual model (Raft) for a complex problem (distributed consensus) against a baseline technique (Paxos). They measured understanding by presenting graduate students with a 1-hour lecture on either system followed by a quiz, finding that Raft quiz scores were 8\% higher than Paxos quiz scores.

We used a similar approach, but adapted to the setting of an online textbook. Rather than present a monolithic lecture followed by a monolithic quiz, we diffused the quiz questions throughout the ownership chapter and the rest of the book. Furthermore, we distinguished between two kinds of quiz questions designed to answer two research questions:

\begin{itemize}[leftmargin=4em,topsep=2pt]
    \item[\textbf{RQ1.}] Does our pedagogy help learners understand ownership \emph{at all}? 
    \item[\textbf{RQ2.}] Does our pedagogy help learners understand ownership \emph{better than before}?
\end{itemize} 

\noindent For RQ1, we asked participants simple comprehension questions about ownership, presented immediately following the book content that is relevant to a given question. These questions determine whether participants can transfer their ownership knowledge to situations similar to the text. Because the questions make references to the permissions model, we cannot establish a score baseline. Therefore we judge the scores in absolute rather than relative terms.

For RQ2, we gave participants a multiple-choice version of the Ownership Inventory. We inserted these questions later in the book after covering the essential prerequisites for a given program. To compare the baseline \trpl{} pedagogy against ours, we ran a kind of temporal A/B test. Participants answered Inventory questions after reading the original \trpl{} content for a few weeks. We then deployed the intervention and continued receiving responses to the Inventory. We quantified the pedagogy's effect based on the resulting change in scores. 

\subsection{Methodology}
\label{sec:evalmethods}

In designing our methodology, we traded off between minimizing the amount of infrastructure required, and maximizing the statistical power of our inferences from data. For example, we did not gather any demographic information about participants. We did not want to dissuade privacy-sensitive people from participating in the experiment (reducing the sample size). Moreover, we did not want to implement the infrastructure necessary to securely manage PII at scale. Nevertheless, given that our participants came over from a popular Rust textbook, we believe that people visiting \trpl{} are reasonably representative of the average Rust learner.

Additionally, the temporal A/B test is logistically simpler but statistically weaker than a traditional population-randomizing A/B test. The traditional setup requires a centralized user database to ensure a reader would not see condition A on their desktop and then accidentally enter condition B on their phone. Instead, our statistical inferences assume that each new participant is sampled from an i.i.d. stream of Rust learners. We discuss this and other trade-offs further in \Cref{sec:threats}. Our methodology was evaluated by Brown's IRB, which determined that the project did not require institutional review due to the study's purpose and safeguards to ensure anonymity.

\subsubsection{Participants}

We recruited participants by advertising in the title page of the official web version of \trpl{}, courtesy of the authors. The advertisement read: ``Want a more interactive learning experience? Try out a different version of the Rust Book, featuring: quizzes, highlighting, visualizations, and more.'' Since this advertisement was put up on \data{November 1, 2022}, our \trpl{} fork has received an average of \data{450} visitors per day, as measured by unique session IDs stored via cookies.

\subsubsection{Materials}

\begin{figure}
    \centering
    \begin{subfigure}[t]{0.49\textwidth}
    \begin{tightexcerpt}
        Consider the permissions in the following program:

        \vspace{0.5em}
        \noindent\begin{center}
            \includegraphics[width=\textwidth]{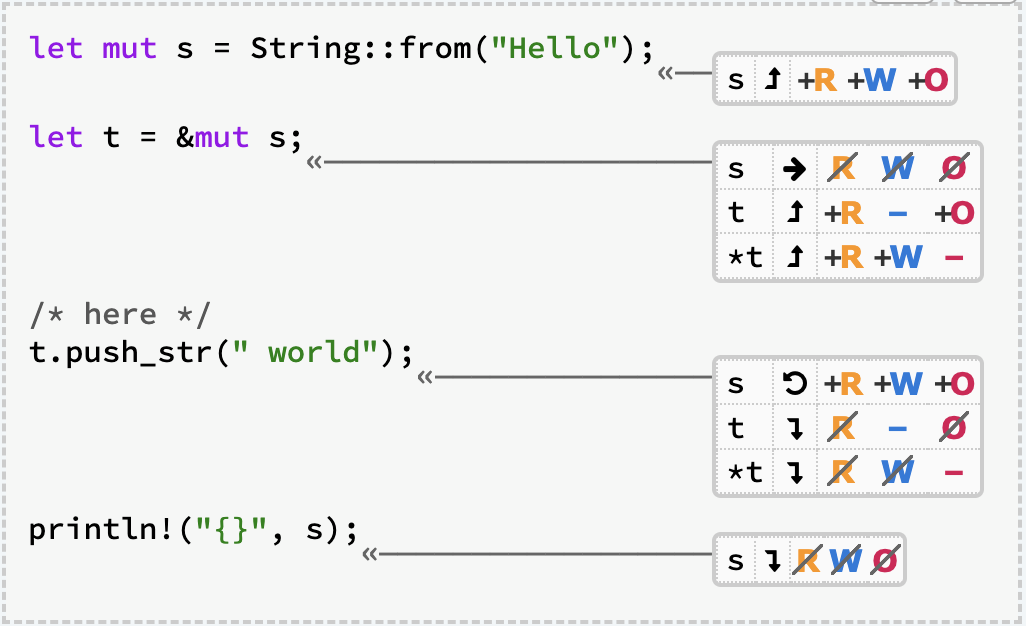}
        \end{center}
        \vspace{0.5em}

        \noindent At the point marked \rs{/* here */}, what are the permissions on the path \rs{s}? Select each permission below, or select ``no permissions'' if the path has no permissions.
        \begin{qchoices}
        \wronganswer $\rpm$
        \wronganswer $\opm$
        \wronganswer $\wpm$
        \rightanswer No permissions
        \end{qchoices}
    \end{tightexcerpt}
    \captionsetup{width=0.95\linewidth}
    \cprotect\caption{A comprehension question (``Analysis state in permissions diagram'') that tests whether a person can correctly interpret a permission diagram. }
    \label{fig:eval-questions-left}
    \end{subfigure}
    \begin{subfigure}[t]{0.50\textwidth}
    \begin{tightexcerpt}
\begin{lstlisting}[basicstyle=\ttfamily\scriptsize]
fn make_separator(user_str: &str) -> &str {
    if user_str == "" {
        let default = "=".repeat(10);
        &default
    } else {
        user_str
    }
}
\end{lstlisting}

    \noindent Assume that the compiler did NOT reject this function. Which (if any) of the following programs would (1) pass the compiler, and (2) possibly cause undefined behavior if executed? Check each program that satisfies both criteria, OR check ``None of these programs'' if none are satisfying.

    \begin{qchoices}
    \rightanswer[add to width=4mm]
\begin{lstlisting}[style=mcq]
let s = make_separator("");
println!("{s}");        
\end{lstlisting}%
    \wronganswer[add to width=-4mm]
\begin{lstlisting}[style=mcq]
let s = 
  make_separator("");
\end{lstlisting}%
    \wronganswer    
\begin{lstlisting}[style=mcq]
println!("{}", 
  make_separator("Hello world!")); 
\end{lstlisting}%
    \wronganswer None of these programs
    \end{qchoices}
    \end{tightexcerpt}
    \cprotect\caption{A multiple-choice Ownership Inventory question (Q2a for the \rs{make_separator} program). The distractors are drawn from common incorrect answers given for the open-ended version of the same question.}
    \label{fig:eval-questions-right}
    \end{subfigure}
    \captionsetup{width=0.95\linewidth}
    \cprotect\caption{Two examples of questions used in the evaluation.}
    \label{fig:eval-questions}
\end{figure}

\noindent We developed \data{11} comprehension questions to cover the content of the ownership chapter. \Cref{sec:allcomprehensionquestions} contains the full text of each question, and \Cref{tab:comprehension-questions} contains a short description of each question. \Cref{fig:eval-questions-left} shows one example --- to test understanding of permission diagrams, we asked participants to infer the permissions for a path at a given point.


We developed \data{24} Ownership Inventory questions based on the \data{8} program in \Cref{tab:open-ended-data}. For each program, we created a close-ended version of Q1, Q2a, and Q3a in \Cref{fig:inventory-question-template} (the justification questions Q2b/Q3b did not translate to the multiple-choice setting). Following the concept inventory methodology, we selected distractors from common misconceptions about the  Inventory programs.
\Cref{fig:eval-questions-right} shows an example of a multiple-choice Ownership Inventory question. The \rs{make_separator} task is an instance of a dangling stack reference. A common incorrect counterexample provided by participants in \Cref{sec:ci} was the snippet \rs{let s = make_separator("")} which creates a dangling pointer, but does not use it. By using that incorrect answer as a distractor, the multiple-choice question is more likely to test for the presence of this misconception.

We spaced out the 24 questions into 4 sets of 6 questions, using the same order as in \Cref{tab:open-ended-data}. Each pair of tasks was embedded into the end of the appropriate chapter such that the cumulative preceding content would cover all necessary features --- Chapter 6, 8, 10, and 17, respectively. Additionally, to ensure participants had access to the relevant documentation on standard library constructs, all the code snippets in the Inventory questions used the same embedded language server technology as described in \Cref{sec:ci-materials}. As an example, the Chapter 6 Inventory questions can be viewed here: \url{https://rust-book.cs.brown.edu/ch06-04-inventory.html}

The quiz widget permits participants to take a quiz multiple times if they answer any question incorrectly. We eliminate any repeat attempts (as determined by the participants' session IDs) from the dataset and only evaluate based on each participant's first attempt on a given quiz.

\subsubsection{Procedure}

On \data{December 13, 2022}, we deployed the Inventory questions to our book. After fixing bugs for a few weeks, on \data{January 3, 2023} we froze the question content and started gathering pre-intervention data on performance with the baseline \trpl{} pedagogy. After \data{45} days, on \data{February 17, 2023} we deployed the initial draft of our new ownership chapter. 

For the next few months, we iterated on the text. This process consisted of reviewing quiz data and user feedback (readers had the option to report broken or confusing text, diagrams, and questions). We did not substantively change the pedagogy during this time, but rather fixed or clarified small issues. Here are two examples of the several changes made in this time:
\begin{itemize}
    \item The original version of our new ownership chapter illustrated borrowing with an example involving vectors. One user alerted us to the fact that \trpl{} does not explain vectors until a later chapter, making the example difficult to understand. In response, we added context explaining the aspects of vectors that are needed for the example.
    \item The original version of the runtime diagram did not visually represent a variable being invalidated upon move, since technically move-invalidation is not part of the Rust runtime (a move is just a copy). But readers frequently complained that they expected moves to be visible in the diagram --- that is, their mental model did not match our conceptual model. In this case, we ultimately agreed with readers that the point was too technical, and added invalidation to the runtime diagram.
\end{itemize}
\noindent A complete list of the changes can be found in the commit log of the GitHub repository for our book\,\citegithub{https://github.com/cognitive-engineering-lab/rust-book}. On \data{June 15, 2023} we froze the text and began gathering post-intervention data. Data collection continued for \data{45} days until \data{July 30, 2023}.

\subsection{Results}

\begin{table}
    \centering
    \begin{tabular}{l|rr}
                                 \textbf{Description} & \textbf{Accuracy} & $\boldsymbol{N}$\\
\midrule
                Difference of stack and heap &     64\% & 2060 \\
                 Aliasing in runtime diagram &     93\% & 2060 \\
                    Moves in runtime diagram &     83\% & 2060 \\
              Defined vs. undefined behavior &     58\% & 1958 \\
                  Compiler error due to move &     84\% & 1958 \\
Dereferencing multiple layers of indirection &     47\% & 1683 \\
               How moves affect deallocation &     95\% & 1683 \\
       Analysis state in permissions diagram &     81\% & 1281 \\
              Why borrows change permissions &     59\% & 1281 \\
 How references can cause undefined behavior &     44\% & 1281 \\
   Compiler error due to overlapping borrows &     87\% & 1180 \\
\bottomrule
\multicolumn{1}{r}{\textbf{Pooled accuracy}} &
\multicolumn{1}{r}{\textbf{72\%}}
\end{tabular}
    \caption{Readers' accuracy on simple comprehension questions about the new ownership pedagogy. Accuracies are average correctness with the number of data points in parentheses. Questions are presented in the order encountered by readers. The full text of the questions is in \Cref{sec:allcomprehensionquestions}.}
    \label{tab:comprehension-questions}
\end{table}

\subsubsection{RQ1: Does our pedagogy help learners understand ownership at all?} 
\Cref{tab:comprehension-questions} shows readers' accuracies on the  comprehension questions. Overall, the pooled accuracy of \data{72\%} shows that readers could mostly understand the basic concepts within our pedagogy.
Readers were able to successfully interpret both runtime and compile-time diagrams (``Aliasing in runtime diagram'' at \data{93\%}, ``Moves in runtime diagram'' at \data{83\%}, ``Analysis state in permissions diagram'' at \data{81\%}). Readers could also identify when the compiler was going to reject a program (``Compiler error due to move'' at \data{84\%}, ``Compiler error due to overlapping borrows'' at \data{87\%}).

However, readers' mediocre performance on a few of the comprehension questions suggests that their understanding may be somewhat shallow. For example, the ownership chapter provides an example program containing a variable \rs{x : &Box<i32>}, and explains that two dereferences like \rs{**x} are needed to access the inner integer. The comprehension question ``Dereferencing multiple layers of indirection'' presents a program that constructs an expression of type \rs{Box<&Box<i32>>} (including a runtime diagram), and asks respondents to determine the number of dereference operations needed to access the inner \rs{i32}. Only \data{47\%} of respondents correctly answer three, suggesting that readers still leave with a somewhat fragile understanding of an essential concept like pointers.
    
\subsubsection{RQ2: Does our pedagogy help learners understand ownership better than before?}

We focus on the first \data{18} Inventory questions, as those questions received enough responses to make statistical inferences. Many readers don't read to the end of the book --- pre-intervention, we only collected \data{77} responses to the last 6 questions versus \data{1,120} for the first 6. This dropout rate is comparable to the 90\%+ dropout rates seen in MOOCs\,\cite{jordan2015mooc}. Additionally, participants answered the first Inventory question on average \data{4} days after answering the last comprehension question. Considerable time elapsed between reading the ownership chapter and taking the Inventory.

First, we analyze the overall Inventory scores for the $N =$ \data{177} (pre) / \data{165} (post) participants who completed the first 18 questions. The average pre-intervention score was \data{48\%} ($\sigma = \data{16\%}$). Notably, the average score on the open-ended Inventory questions in \Cref{sec:ci} was \data{41\%} (which should be more difficult than equivalent multiple-choice questions), showing that the quantitative results of the formative study reasonably generalized to a larger sample. The average post-intervention score was \data{57\%} ($\sigma = \data{15\%}$). Using a two-tailed $t$-test, the difference is statistically significant ($p < 0.001$). The normalized effect size as measured by Cohen's $d$ is $\data{0.56}$. Therefore, the pedagogy had a statistically significant positive effect (\data{+9\%}) on overall Inventory performance. Additionally, the results confirm that Inventory questions are substantially harder than the comprehension questions.


\begin{table}
\begin{tabular}{ll|rr|rr|r|r|r}
    \textbf{Task} & \textbf{Q.} & \textbf{Before} & $\boldsymbol{N}$ & \textbf{After} & $\boldsymbol{N}$ & \textbf{Effect} & $\boldsymbol{d}$ &  $\boldsymbol{p}$ \\
\midrule
\verb|make_separator| &       Q2 &   33\% &      1120 &  40\% &      660 &   +7\% & 0.14 &  \textbf{0.007} \\
\verb|make_separator| &       Q3 &   57\% &      1120 &  62\% &      660 &   +5\% & 0.11 &  \textbf{0.024} \\
\verb|get_or_default| &       Q1 &   56\% &      1120 &  66\% &      660 &  +10\% & 0.21 & \textbf{<0.001} \\
\verb|get_or_default| &       Q2 &   10\% &      1120 &  16\% &      660 &   +6\% & 0.19 & \textbf{<0.001} \\
  \verb|remove_zeros| &       Q3 &   35\% &       629 &  52\% &      470 &  +17\% & 0.34 & \textbf{<0.001} \\
       \verb|reverse| &       Q2 &   28\% &       629 &  40\% &      470 &  +13\% & 0.27 & \textbf{<0.001} \\
       \verb|reverse| &       Q3 &   21\% &       629 &  33\% &      470 &  +13\% & 0.29 & \textbf{<0.001} \\
      \verb|find_nth| &       Q1 &   86\% &       314 &  90\% &      374 &   +4\% & 0.12 &           0.099 \\
      \verb|find_nth| &       Q2 &   16\% &       314 &  23\% &      374 &   +7\% & 0.18 &  \textbf{0.018} \\
      \verb|find_nth| &       Q3 &   27\% &       316 &  34\% &      374 &   +7\% & 0.15 &  \textbf{0.041} \\
   \verb|apply_curve| &       Q2 &   41\% &       452 &  58\% &      374 &  +17\% & 0.34 & \textbf{<0.001} \\
\bottomrule
    \multicolumn{6}{r}{\textbf{Pooled significant effect:}} &
    \multicolumn{1}{r}{\textbf{+10\%}} &
    \multicolumn{1}{r}{\textbf{0.22}}
\end{tabular}
\caption{Effects of the permissions pedagogy for readers' accuracy on Ownership Inventory questions. Questions are presented in the order encountered by readers. Only effects with $p < 0.15$ are included here, with $p < 0.05$ in bold.}
\label{tab:mcq-inventory-data}
\end{table}

Second, we analyze the intervention's effect on each Inventory question individually. The intervention had a statistically significant effect on \data{10/18} questions. \Cref{tab:mcq-inventory-data} shows the size of these effects, including almost-significant effects. Overall, the pooled significant effect was \data{10\%} or $d = \data{0.22}$ (note that the question-level effect size is smaller than the quiz-level effect size due to the higher per-question variance). Between the different types of questions, the intervention primarily affected performance on questions about undefined behavior (Q2) and fixing a type error (Q3) moreso than identifying a type error (Q1). For example with \rs{make_separator} Q2, the \data{+7\%} effect corresponds to an \data{8\%} decrease in the incorrect response of ``does not have counterexamples''. Conversely, for \rs{reverse} and \rs{apply_curve} Q2, the \data{+13\%}/\data{+17\%} effects correspond to participants answering correctly that these functions are safe and do not have counterexamples.



\subsection{Discussion}
The results on the comprehension questions show that our pedagogy is comprehensible to the average Rust learner. That is notable \textit{per se}, as we have no control over the learner population, many of whom  come with no experience in relevant areas like systems or functional programming. 

The results on the Inventory questions show that the effect of our pedagogy is statistically significant with an effect size of $d = \data{0.56}$. 
For reference, according to the meta-analysis of education research by \citet{Hattie2008-xf}, the average effect of educational interventions on learning outcomes is $d = 0.40$. \citet[p.~17]{Hattie2008-xf} argues that ``the effect size of 0.40 sets a level where the effects of innovation enhance achievement
in such a way that we can notice real-world differences, and this should be a benchmark of such real-world change.'' Therefore, we interpret our results as saying that the new ownership pedagogy is a substantive step in the right direction. But a post-intervention average of 57\% clearly demonstrates that we have not ``closed the book'' on the challenge of teaching ownership types.

\subsection{Threats to Validity}
\label{sec:threats}

Given the large scope of this experiment, we considered several threats to validity in its design.

\subsubsection{Construct validity}

This experiment assumes that the Ownership Inventory is a valid instrument to measure a person's understanding of ownership. To that end, we designed the Inventory such that the situations reflect common ownership problems (by weighting based on StackOverflow), and such that the questions reflect each stage of reasoning about ownership (based on our formative study). However, future work should validate the extent to which performance on the Inventory correlates to performance in solving ownership problems in practice.

\subsubsection{Internal validity} 

The setting of an online textbook provides the benefit of scale, but it also poses methodological challenges due to lack of controls. One such threat is the uncontrolled quizzing environment. A reader could augment their problem-solving with external aids like a friend, a compiler, a Google search, a large language model, and so on. Participants could also be influenced by learning material outside of the book, such as the official \trpl{} or Rust-related YouTube videos. To combat this threat, we explicitly instructed participants to not use external resources while solving quiz problems, and the quiz widget takes over the browser tab while taking a quiz. Moreover, we assume that the average participant will be a good actor --- our readers are taking these quizzes for their own edification, not to get paid by us or to get a good grade. Gathering enough data should turn bad actors into noise.

Another threat is the uncontrolled assignment to experimental conditions. We chose not to perform a randomized-controlled trial for the reasons discussed in \Cref{sec:evalmethods}. However, it is possible that temporal correlations in readership could have affected our results. For example, if all the C++ engineers at one company decided to start learning Rust at the same time, then average scores would likely go up in that window of time compared to the average in the limit.

A final threat is teaching to the test. Unlike us, the \trpl{} authors were not aware of the Ownership Inventory when they wrote the book. At the extreme, if our pedagogy taught the exact answers to Inventory questions, then Inventory scores would not be a useful measure of ownership understanding.
At the same time, part of the point of our experiment is exactly to teach to the test! For example, the Inventory is intentionally designed to measure understanding of \ub{}, and in turn we intentionally designed our new pedagogy to explain \ub{}. Like any well-meaning educator, we sought a balance. The Inventory materials do not appear anywhere in the revised text. But we do, for instance, walk through an example of how iterator invalidation causes \ub{}, which is similar to the \rs{remove_zeros} problem.

\subsubsection{External validity}

Conditioned on construct and internal validity, our results should reasonably generalize to the larger population of Rust learners.  \trpl{} is the official Rust textbook for the community, so its readers should be a representative cross-section of the broader Rust ecosystem.

\section{Related Work}

In response to the reports of learners' struggles with ownership\,\cite{zeng2018,zhu2022,fulton2021,rustsurvey2020}, researchers have explored several ways to help Rust users deal with ownership. For instance, \citet{coblenz2022bronze} showed that garbage collection can help users avoid ownership issues and thereby complete a coding task more quickly. 

More directly relevant to our work are \emph{ownership visualizations}. \citet{dominik2018} and \citet{blaser2019} developed a tool that visualizes a graph of the outlives-constraints generated by the Rust compiler. They did not evaluate the human factors of their tools, and we believe their visualization would be more appropriate for aiding compiler engineers than learners.
%
\citet{almeida2022rustviz} created RustViz, a visualization format for ownership annotations on a Rust program. In terms of pedagogy, RustViz's premise is that the key challenge with ownership is that ``the user must learn to mentally simulate the logic of the borrow checker''. Our pedagogy is based more on connecting Rust's static and dynamic semantics, which we show in our formative study is a more serious problem for Rust learners. 
In terms of implementation, RustViz diagrams are constructed by hand using a DSL, while we automatically generate our diagrams from the compiler. 
In terms of evaluation, Almeida et al. deployed RustViz in a classroom, finding that students responded to a Likert item that the visualizations were ``helpful in terms of improving their understanding of ownership.'' Our evaluation goes further to quantify the effect of our pedagogy on learning outcomes.

Our runtime diagram is similar to program state visualizers in prior work --- see \citet{sorva2013} for a survey. In particular, our work is similar to C runtime visualizations\,\cite{ishizue2018pvc,egan2021seec,taylor2023debugc}. In the same vein, our findings about misconceptions of undefined behavior and memory safety are consistent with prior work on teaching C. For instance, \citet{lam2022secure} found in a study of undergraduates who had taken a computer systems course that ``many students displayed little knowledge or had misunderstandings about memory and memory layout'' and would simply say ``something bad'' happens when unsafe operations occur.

Our work continues a line of CS education research about conceptual models. \citet{bayman1988models} first showed that an appropriate conceptual model for BASIC could help students ``develop fewer misconceptions [...] and perform better on transfer tests.'' \citet{dubolay1986} coined the term ``notional machine'' for conceptual models specifically of a language's dynamic semantics, which has received renewed focus in recent years\,\cite{dickson2020notional}. 
Our work differs from most research on notional machines by focusing equally on a conceptual model of \emph{static} semantics.

Our work also intersects with a line of programming language research on the human factors of type systems and functional languages. Most prior work has focused on algorithms for identifying the root cause of confusing type inference errors\,\cite{wand1986typeerrors,chitil2001error,zhang2014error}. Recent work has broadened scope to develop theories about how programmers read functional programs\,\cite{marceau2011values}, leverage the type system during development\,\cite{lubin2021fp}, and solve problems with higher-order functions\,\cite{rivera2022hof}. 

This paper focuses on ownership types as they are implemented in Rust, but ownership types have taken many forms in prior work\,\cite{Clarke2013}. For instance, early systems of ownership focused on ensuring uniqueness of access to data by checking for dominance in the alias graph\,\cite{clarke1998ownership}. Later systems relaxed this constraint by permitting temporary borrowing of data, both mutably\,\cite{boyland2001burying,aldrich2002alias} and immutably\,\cite{dietl2012universe,ostlund2008own}. The connection between ownership and permissions has been well-established within formal models such as fractional permissions\,\cite{boyland2003frac} and $F_{\text{own}}$\,\cite{krishnaswami2005perm}.

\section{General Discussion}
\label{sec:conclusion}

Future programming languages will undoubtedly have increasingly complex type systems. Rust is the language \textit{du jour}, so this work focused on ownership types. But the next popular language could bring a renewed emphasis to any existing line of PL research: refinement types, session types, or even theorem proving. Effective transfer of these technologies will require pedagogies that do not expect learners to come equipped with Ph.D.-level knowledge of programming languages, mathematics, and Greek. While our immediate goal in this work was to make ownership types more understandable, our broader goal was to explore the viability of different techniques for improving PL pedagogy. In this section, we will briefly reflect on lessons learned.

First, to develop a metric for understanding of ownership, we created a concept inventory by combining data from StackOverflow with a formative study of Rust learners. StackOverflow works for popular languages like Rust, but is less useful for niche languages. Human factors research on niche languages can instead consider using telemetry from developer interactions as has been explored for Racket\,\cite{marceau2011racket} and Coq\,\cite{ringer2020replica}. The concept inventory is an idea that could easily be reused in the context of other languages. Inventories can serve as communal benchmarks for progress in education research, like how datasets of programs serve as benchmarks for performance in compiler research. 

Second, to develop a conceptual model for ownership, we carefully selected a level of abstraction that was concrete enough to explain relevant phenomena like \ub{}, while abstract enough to avoid unnecessary details. We leveraged the rich prior work on distilling the Rust type system into a small, explainable set of mechanisms, especially the Oxide\,\cite{weiss2019oxide} and Polonius\,\cite{polonius} models. However, PL research usually distills type systems to permit formal reasoning, such as a soundness proof. An open question is how to distill type systems for didactic reasoning, that is, to help learners acquire a conceptual model valid for common tasks. For example, one of our principles was that our model must be encodable in a concise visual representation, which is not a property usually expected of standard PL research. Future work can investigate the properties of semantics that make them more or less explainable.

Finally, to evaluate the efficacy of our pedagogy, we publicly deployed our textbook and compared pre/post-intervention scores on the Ownership Inventory. Collecting telemetry from quizzes in online learning resources is a readily applicable strategy for other contexts. Learners \emph{want} to take quizzes to engage with the content they are reading. Temporal A/B testing offers a lightweight method for evaluating content changes without sophisticated infrastructure. We encourage anyone interested in programming language learning to try out our methodology. To that end, we have open-sourced our frontend quiz plugin and our backend telemetry system at: \anon{\url{https://github.com/cognitive-engineering-lab}}

\begin{acks}
The authors are immensely grateful to Niko Matsakis and Amazon. They provided both the encouragement and the funding to initiate this project, and supplied additional emergency funding when our first grant application fell through because we were studying Rust instead of C++. We thank Carol Nichols for taking a leap of faith in allowing us to advertise in \trpl{}; this was essential for driving traffic to the experiment. Later parts of this work are partially supported by the US NSF under Grant No.~2319014.

\end{acks}

\bibliography{bibs/wc,bibs/misc}

\newpage
\appendix
\section{Appendix}

 \subsection{Ownership Inventory Snippets}
\label{sec:all-inventory-snippets}

\begin{description}
\item[\texttt{make\_separator}] \

\begin{lstlisting}
/// Makes a string to separate lines of text, 
/// returning a default if the provided string is blank
fn make_separator(user_str: &str) -> &str {
  if user_str == "" {
    let default = "=".repeat(10);
    &default
  } else {
    user_str
  }
}
\end{lstlisting}

\item[\texttt{get\_or\_default}] \

\begin{lstlisting}
/// Gets the string out of an option if it exists,
/// returning a default otherwise
fn get_or_default(arg: &Option<String>) -> String {
  if arg.is_none() {
      return String::new();
  }
  let s = arg.unwrap();
  s.clone()
}
\end{lstlisting}

\item[\texttt{find\_nth}] \

\begin{lstlisting}
/// Returns the n-th largest element in a slice
fn find_nth<T: Ord + Clone>(elems: &[T], n: usize) -> T {
  elems.sort();
  let t = &elems[n];
  return t.clone();
}    
\end{lstlisting}

\item[\texttt{remove\_zeros}] \

\begin{lstlisting}
/// Removes all the zeros in-place from a vector of integers.
fn remove_zeros(v: &mut Vec<i32>) {
  for (i, t) in v.iter().enumerate().rev() {
    if *t == 0 {
      v.remove(i);
    }
  }
}    
\end{lstlisting}

\item[\texttt{get\_curve}] \

\begin{lstlisting}
struct TestResult {
  /// Student's scores on a test
  scores: Vec<usize>,

  /// A possible value to curve all sores
  curve: Option<usize>
}
impl TestResult {  
  pub fn get_curve(&self) -> &Option<usize> { 
    &self.curve 
  }

  /// If there is a curve, then increments all 
  /// scores by the curve
  pub fn apply_curve(&mut self) {
    if let Some(curve) = self.get_curve() {
      for score in self.scores.iter_mut() {
        *score += *curve;
      }
    }
  }
}    
\end{lstlisting}

\item[\texttt{reverse}] \

\begin{lstlisting}
/// Reverses the elements of a vector in-place
fn reverse(v: &mut Vec<i32>) {
  let n = v.len();
  for i in 0 .. n / 2 {
    std::mem::swap(&mut v[i], &mut v[n - i - 1]);
  }
}    
\end{lstlisting}

\item[\texttt{concat\_all}] \

\begin{lstlisting}
/// Adds the string `s` to all elements of 
/// the input iterator
fn concat_all(
  iter: impl Iterator<Item = String>,
  s: &str
) -> impl Iterator<Item = String> {
  iter.map(move |s2| s2 + s)
}
\end{lstlisting}

\item[\texttt{add\_displayable}] \

\begin{lstlisting}
/// Adds a Display-able object into a vector of 
/// Display trait objects
use std::fmt::Display;
fn add_displayable<T: Display>(
  v: &mut Vec<Box<dyn Display>>, 
  t: T
) {
  v.push(Box::new(t));
}    
\end{lstlisting}
\end{description}

\newpage
\subsection{Permissions Model Proofs}
\label{sec:proofs}

\begin{theorem}
$\accesserror{G} \vdash \permfail{G}$
\end{theorem}
\begin{proof}
By cases on the derivation of $\accesserror{G}$.

\begin{description}
\item[Case: \cref{tr:b-borrow-conflict}:] Let $p, \omega, I$ such that $\liveat{\borrow{p}{\omega}}{I}$ and $\invalidat{\borrow{p}{\omega}}{I}$. Case on the derivation of $\msf{invalidated~at}$.

\begin{description}
\item[Sub-case: \cref{tr:b-read-invalid}:] Let $p'$ such that $\readat{p'}{I}$ and $p \conflicts p'$. Assume that $\omega = \uniq$. Then we can derive $\permfail{G}$:

$$
\infer*[right=(\cref{tr:p-fail})]
    {\infer*[right=(\cref{tr:p-needs-read})]
        {\readat{p'}{I}}
        {\needsat{p'}{\rpm}{I}} \\    
     \infer*[right=(\cref{tr:p-missing-read})]
        {\liveat{\borrow{p}{\uniq}}{I} \\ 
         p \conflicts p'}
        {\missingat{p'}{\rpm}{I}}}
    {\permfail{G}}
$$

\item[Sub-case: \cref{tr:b-write-invalid}:] Let $p'$ such that $\writtenat{p'}{I}$ and $p \conflicts p'$.  Then we can derive $\permfail{G}$:

$$
\infer*[right=(\cref{tr:p-fail})]
    {\infer*[right=(\cref{tr:p-needs-write})]
        {\writtenat{p'}{I}}
        {\needsat{p'}{\wpm}{I}} \\ 
     \infer*[right=(\cref{tr:p-missing-write-own})]
        {\liveat{\borrow{p}{\omega}}{I} \\ 
         p \conflicts p'}
        {\missingat{p'}{\wpm}{I}}}
    {\permfail{G}}    
$$

\item[Sub-case: \cref{tr:b-move-invalid}:] Let $p'$ such that $\movedat{p'}{I}$ and $p \conflicts p'$. Then we can derive $\permfail{G}$:

$$
\infer*[right=(\cref{tr:p-fail})]
    {\infer*[right=(\cref{tr:p-needs-own})]
        {\movedat{p'}{I}}
        {\needsat{p'}{\opm}{I}} \\ 
     \infer*[right=(\cref{tr:p-missing-write-own})]
        {\liveat{\borrow{p}{\omega}}{I} \\ 
         p \conflicts p'}
        {\missingat{p'}{\opm}{I}}}
    {\permfail{G}}    
$$

\end{description}

\item[Case: \cref{tr:b-move-conflict}:] Let $p, I$ such that $\movedbefore{p}{I}$ and $\readat{p}{I}$. Then we can derive $\permfail{G}$:

$$
\infer*[right=(\cref{tr:p-fail})]
    {\infer*[right=(\cref{tr:p-needs-read})]
        {\readat{p}{I}}
        {\needsat{p}{\rpm}{I}} \\
    \infer*[right=(\cref{tr:p-moved-no-permissions})]
        {\movedbefore{p}{I}}
        {\missingat{p}{\rpm}{I}}}
    {\permfail{G}}
$$
\end{description}
\end{proof}

\newpage
\subsection{Comprehension Questions}
\label{sec:allcomprehensionquestions}

\begin{description}
\item[Difference of stack and heap] \

\begin{tightexcerpt}
Which of the following best describes the difference between the stack and the heap?
\begin{qparachoices}
    \rightanswer The stack holds data associated with a specific function, while the heap holds data that can outlive a function.
    \wronganswer The stack can hold pointers to data stored on the heap, while the heap only holds data without pointers.
    \wronganswer The stack holds immutable data, while the heap holds mutable data.
    \wronganswer The stack holds copyable data, while the heap holds uncopyable data.
\end{qparachoices}
\end{tightexcerpt}

\item[Aliasing in runtime diagram] \

\begin{tightexcerpt}
Consider the execution of the following snippet, with the final state shown:

\includegraphics[width=0.3\linewidth]{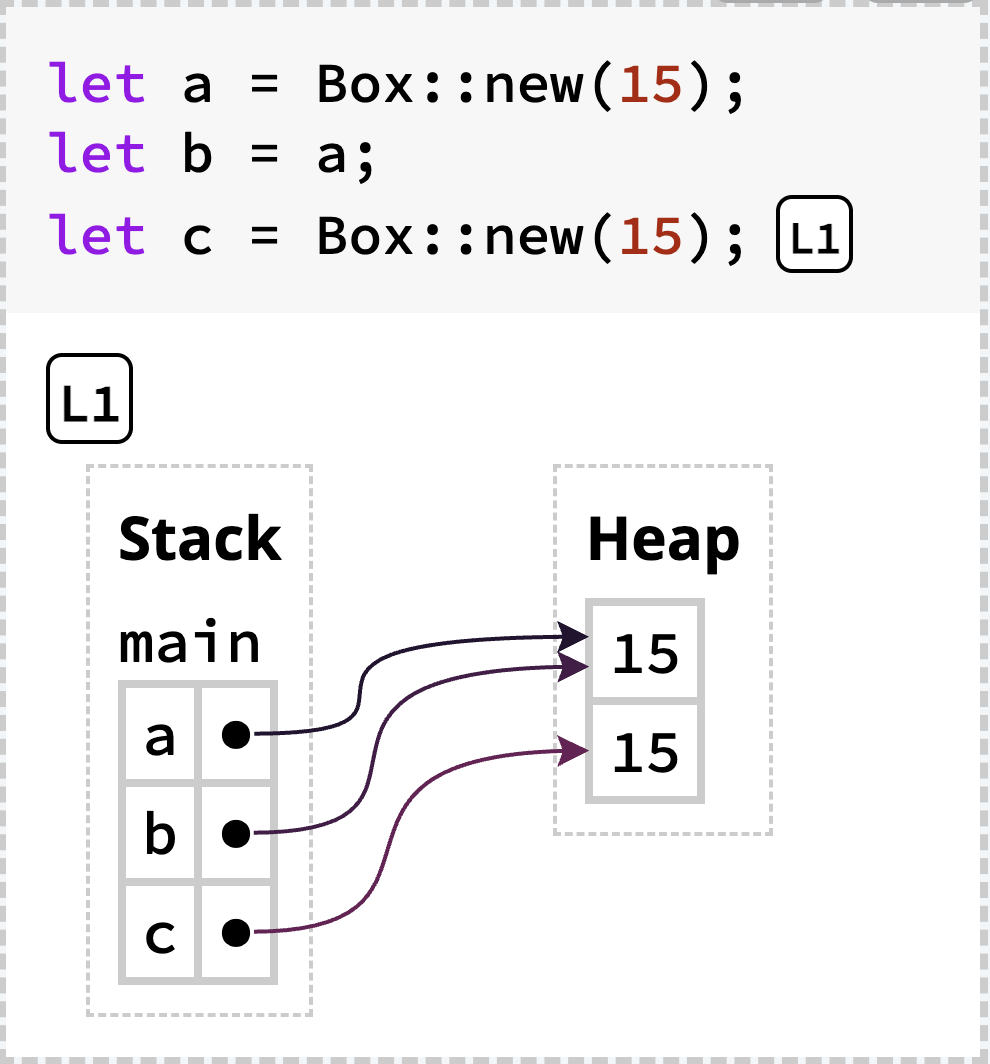}

In the final state, how many copies of the number \rs{15} live anywhere in memory? Write your answer as a digit, such as 0 or 1.
\begin{qparachoices}
\rightanswer 2
\end{qparachoices}
\end{tightexcerpt}

\item[Moves in runtime diagram] \

\begin{tightexcerpt}
Consider the execution of the following snippet, with an intermediate state shown:

\includegraphics[width=0.48\linewidth]{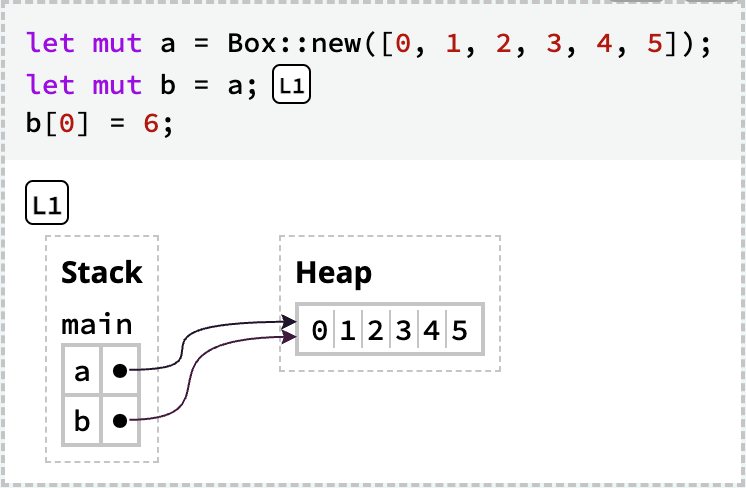}

In the final state of the snippet, what is the value of \rs{a[0]}? Write your answer as a digit, e.g. 0 or 1.
\begin{qparachoices}
\rightanswer 6
\end{qparachoices}
\end{tightexcerpt}

\newpage
\item[Defined vs. undefined behavior] \

\begin{tightexcerpt}
Which of the following is NOT a kind of undefined behavior?
\begin{qparachoices}
    \rightanswer Having a pointer to freed memory in a stack frame
    \wronganswer Using a pointer that points to freed memory
    \wronganswer Freeing the same memory a second time
    \wronganswer Using a non-boolean value as an `if` condition
\end{qparachoices}
\end{tightexcerpt}

\item[Compiler error due to move] \

\begin{tightexcerpt}
Determine whether the program will pass the compiler. If it passes, write the expected output of the program if it were executed.

\begin{lstlisting}
fn add_suffix(mut s: String) -> String {
  s.push_str(" world");
  s
}

fn main() {
  let s = String::from("hello");
  let s2 = add_suffix(s);
  println!("{}", s2);
}
\end{lstlisting}

\begin{qparachoices}
    \rightanswer It does compile, with output ``hello world''
    \wronganswer It does not compile
\end{qparachoices}
\end{tightexcerpt}

\newpage
\item[Dereferencing multiple layers of indirection] \

\begin{tightexcerpt}
Consider the following program, showing the state of memory after the last line:

\includegraphics[width=0.4\textwidth]{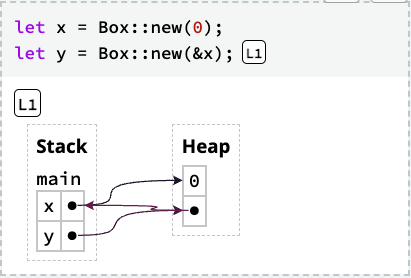}

If you wanted to copy out the number \rs{0} through \rs{y}, how many dereferences would you need to use? Write your answer as a digit. For example, if the correct expression is \rs{*y}, then the answer is 1. 

\begin{qparachoices}
    \rightanswer 3
\end{qparachoices}
\end{tightexcerpt}

\item[How moves affect deallocation] \

\begin{tightexcerpt}
Consider the following program, showing the state of memory after the last line:

\includegraphics[width=0.4\textwidth]{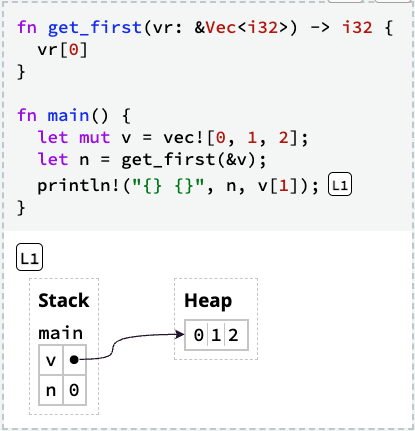}

Which of the following best explains why \rs{v} is not deallocated after calling \rs{get_first}?

\begin{qparachoices}
    \rightanswer \rs{vr} is a reference which does not own the vector it points to
    \wronganswer \rs{vr} is not mutated within \rs{get_first}
    \wronganswer \rs{get_first} returns a value of type \rs{i32}, not the vector itself
    \wronganswer \rs{v} is used after calling \rs{get_first} in the \rs{println}
\end{qparachoices}
\end{tightexcerpt}

\newpage
\item[Analysis state in permissions diagram] \

\begin{tightexcerpt}
Consider the permissions in the following program:

\includegraphics[width=0.6\textwidth]{figs/simple-eval-question.png}

At the point marked \rs{/* here */}, what are the permissions on the path \rs{s}? Select each permission below, or select ``no permissions'' if the path has no permissions.
\begin{qchoices}
\wronganswer $\rpm$
\wronganswer $\opm$
\wronganswer $\wpm$
\rightanswer No permissions
\end{qchoices}
\end{tightexcerpt}

\newpage
\item[Why borrows change permissions] \

\begin{tightexcerpt}
Consider the permissions in the following program:

\includegraphics[width=0.7\textwidth]{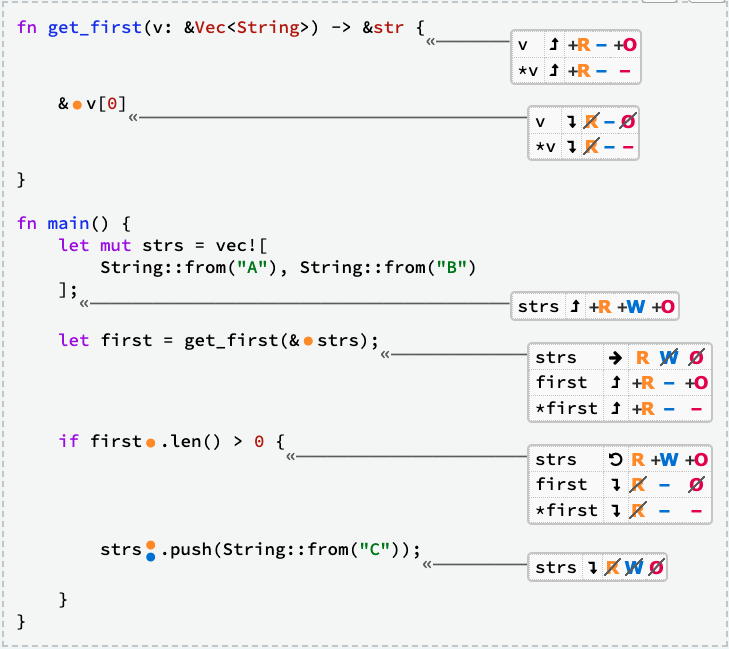}

Which of the following best explains why \rs{strs} loses and regains write permissions?
\begin{qparachoices}
\rightanswer \rs{get_first} returns an immutable reference to data within \rs{strs}, so \rs{strs} is not writable while \rs{first} is live
\wronganswer \rs{strs} is not writable while the immutable reference \rs{&strs} passed to \rs{get_first} is live
\wronganswer \rs{strs} does not need write permissions until the \rs{strs.push(..)} operation, so it only regains write permissions at that statement
\wronganswer Because \rs{first} refers to \rs{strs}, then \rs{strs} can only be mutated within a nested scope like the if-statement
\end{qparachoices}
\end{tightexcerpt}

\newpage
\item[How references can cause undefined behavior] \

\begin{tightexcerpt}
Consider this unsafe program:

\includegraphics[width=0.43\textwidth]{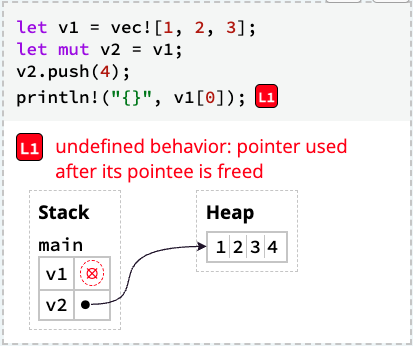}

Which of the following best describes how undefined behavior occurs in this program?
\begin{qparachoices}
\rightanswer \rs{v1[0]} reads \rs{v1}, which points to deallocated memory
\wronganswer \rs{v2} owns the vector data on the heap, while \rs{v1} does not
\wronganswer \rs{v1} has been moved into \rs{v2} on line 2
\wronganswer \rs{v1} has its pointer invalidated by the \rs{push} on line 3
\end{qparachoices}
\end{tightexcerpt}

\item[Compiler error due to overlapping borrows] \

\begin{tightexcerpt}
Determine whether the program will pass the compiler. If it passes, write the expected output of the program if it were executed.

\begin{lstlisting}
fn main() {
  let mut s = String::from("hello");
  let s2 = &s;
  let s3 = &mut s;
  s3.push_str(" world");
  println!("{s2}");
}
\end{lstlisting}

\begin{qparachoices}
    \wronganswer It does compile, with output \_\_\_\_
    \rightanswer It does not compile
\end{qparachoices}
\end{tightexcerpt}

\end{description}

\end{document}